%% file: main.tex
\documentclass[journal,twoside,web]{ieeecolor}
\usepackage{silence}
\input{Files/preamble}

\newcommand{\lha}[1]{\textcolor{black}{#1}} 
\newcommand{\lhb}[1]{\textcolor{black}{#1}} 
\newcommand{\lhc}[1]{\textcolor{black}{#1}} 
\newcommand{\lhe}[1]{\textcolor{black}{#1}} 
\newcommand{\lhf}[1]{\textcolor{black}{#1}} 
\newcommand{\lhg}[1]{\textcolor{black}{#1}} 
\newcommand{\lhh}[1]{\textcolor{black}{#1}} 

\title{
Decentralised Plug-and-Play Stability Conditions for AC Grids---Part I}

\author{Liam Hallinan and Ioannis Lestas %
    \thanks{L. Hallinan and I. Lestas are with the Department of Engineering, University of Cambridge, Trumpington Street, Cambridge, CB2 1PZ, United Kingdom. Emails:
        {\tt\small <lh706, icl20>@cam.ac.uk}.}%
}

\begin{document}
\maketitle

\begin{abstract}
    The rise of renewable generation in AC power systems calls for small-signal stability assessment methods that do not rely on centralised, system-wide models. In this two-part paper, we present a decentralised stability framework based on \lhh{frequency-domain} \lhg{conditions} imposed on local subsystems (many admitting equivalent graphical interpretations) which collectively imply \lhg{the stability of} the entire grid. The framework generalises many previously reported results, with added flexibility \lhf{and plug-and-play compatibility} provided by the additional degrees of freedom of the \lhg{imposed conditions}, the ability to mix conditions across frequency ranges, conditions on extended subsystems that couple buses to their connecting lines and loads, \lhg{and a method to incorporate unstable subsystems}. Detailed, heterogeneous device and line models are accommodated in the frequency domain \lhg{using} impedance, admittance, \lhf{and power-flow models}, making the framework suitable as a foundation for grid codes that enable plug-and-play functionality in inverter-dominated grids. \lhf{Part I focuses on stable impedance and admittance representations} \lhg{and Part II extends the results to unstable subsystems, which are shown to be fundamentally present in many grid models}. The results \lhg{of Part I} are validated on a modified 9-bus case study, demonstrating the reduced conservatism \lhh{and significance} of the proposed approach. 
\end{abstract}

\section{Introduction}
\IEEEPARstart{P}{ower} systems are undergoing a profound transformation, driven by global decarbonisation targets, volatile fossil fuel markets, and rapid technological advancements.
As synchronous generation is replaced by inverter-based resources, the rotational inertia inherent to large turbines is disappearing, with significant consequences for how the future grid is stabilised and controlled \cite{milano_FoundationsChallenges_18}. 
Moreover, the highly distributed nature of the evolving grid means that classical small-signal techniques, such as eigenvalue analysis applied to fixed system-wide models, are no longer scalable. 

These challenges have motivated research into stability assessment methods that are scalable and compatible with decentralised grid operation \cite{spanias_SystemReference_19, watson_ScalableControl_21, strehle_UnifiedPassivityBased_21,  dey_PassivityBasedDecentralized_23, devane_PrimaryFrequency_17, 
vorobev_DecentralizedStability_19, chen_ExtendedFrequencyDomain_24, huang_GainPhase_24, huang_GeometricDecentralized_25, leng_GeometricDecentralized_26, gorbunov_DynamicPassivity_26, cifelli_DecentralizedSmall_25, haberle_DecentralizedParametric_25}.
Rather than performing a macroscopic centralised study, such approaches conclude stability by imposing constraints on the input--output properties or control parameters of individual grid subsystems. \lhf{These constraints must be sufficiently non-conservative to permit decentralised implementation while still guaranteeing network-wide stability}. In practice, these constraints take the form of an interconnection standard or grid code that every device must meet. If designed well, such standards can guarantee a level of plug-and-play compatibility, meaning that device-level controllers need not be redesigned whenever the network is altered.

However, previous work in this area either imposes overly conservative conditions on local subsystems, such as passivity \cite{spanias_SystemReference_19,watson_ScalableControl_21, strehle_UnifiedPassivityBased_21,  dey_PassivityBasedDecentralized_23} (which is difficult to satisfy at low frequencies \cite{chen_LimitationsUsing_24}), or simplifies the network dynamics by excluding voltage interactions \cite{devane_PrimaryFrequency_17}, ignoring frequency synchronisation \cite{watson_ScalableControl_21, strehle_UnifiedPassivityBased_21}, assuming lossless networks \cite{devane_PrimaryFrequency_17, chen_ExtendedFrequencyDomain_24} or homogeneous line parameters \cite{haberle_DecentralizedParametric_25,huang_GeometricDecentralized_25, leng_GeometricDecentralized_26}, or neglecting high-frequency bus dynamics \cite{haberle_DecentralizedParametric_25}. In addition, many approaches require centralised knowledge of the network parameters and interconnection structure \cite{huang_GainPhase_24, haberle_DecentralizedParametric_25, huang_GeometricDecentralized_25, leng_GeometricDecentralized_26, cifelli_DecentralizedSmall_25, gorbunov_DynamicPassivity_26}, meaning they are not robust to topological changes and fall short of true plug-and-play operation.

\lhf{Furthermore, the definition of a subsystem and the corresponding network coupling variables is not unique, with previous studies employing impedance-based \cite{spanias_SystemReference_19,watson_ScalableControl_21, strehle_UnifiedPassivityBased_21,dey_PassivityBasedDecentralized_23}, admittance-based \cite{vorobev_DecentralizedStability_19, chen_ExtendedFrequencyDomain_24, huang_GainPhase_24, huang_GeometricDecentralized_25, leng_GeometricDecentralized_26, gorbunov_DynamicPassivity_26}, and power-based \cite{dey_PassivityBasedDecentralized_23, cifelli_DecentralizedSmall_25, haberle_DecentralizedParametric_25} models of the bus subsystems. 
However, these representations may exhibit unstable open-loop poles in many common scenarios, further complicating the application of decentralised stability methods. These instabilities are often overlooked in the literature, an issue we address directly in Part II.}

To address these challenges, this two-part paper proposes a comprehensive small-signal stability certification framework for AC grids that \lhf{accounts for unstable subsystem poles,} generalises many previous results in the literature, and extends them to alternative subsystem representations, offering reduced conservatism. The framework takes the form of a set of quadratic constraints on the input--output properties of local grid subsystems in the frequency domain, which collectively ensure that the multivariable Nyquist criterion holds and that the grid is therefore stable \cite{desoer_GeneralizedNyquist_80}.
Although similar certificates were recently demonstrated for DC grids \cite{laib_DecentralizedStability_23}, extending the methods to AC networks \lhf{and deriving sufficiently non-conservative constraints} is more involved due to their MIMO nature and the need to regulate frequency dynamics. 

The proposed approach builds on integral quadratic constraint (IQC) theory \cite{megretski_SystemAnalysis_97}, extending earlier results for heterogeneous network systems \cite{laib_DecentralizedStability_23, lestas_NetworkStability_11, lestas_LargeScale_12, pates_ScalableDesign_17, khong_UnifyingFramework_16}. The results here are presented in a Nyquist-based framework so that homotopies linking alternative system configurations can be constructed \lhg{--- a key feature in the analysis that reduces the conservatism of the formulated specifications}. 

The advantageous features and contributions of the framework \lhf{developed across both parts of this paper} are as follows:
\begin{itemize}
    \item Constraints can be imposed on multiple representations of the bus subsystems --- the bus impedance (mapping input current to voltage), the bus admittance (its inverse), \lhf{or a power representation (mapping active and reactive power to phase angle and voltage magnitude)} --- and, within a given representation, different criteria can be mixed across regions of the frequency domain, reducing conservatism relative to classical approaches.

    \item In the impedance \lhf{and power} representations, the sparsity of the network can be further exploited to impose constraints on extended decentralised subsystems comprising a bus together with all the lines and loads connected to it. Bus-level and extended-subsystem conditions can then be mixed across different frequency ranges, further reducing conservatism.

    \item If designed appropriately, the decentralised constraints can be verified locally without centralised knowledge of the network parameters or topology. They therefore scale with network size and can enable true plug-and-play integration, with minimal disruption to non-local subsystems when a new bus or line is added.

    \item \lhf{When open-loop unstable poles are present in either the impedance or power representations, a dynamic loop transformation (a frequency-dependent change of the port variables) can be used to create a stable hybrid impedance/power representation, allowing plug-and-play compatible conditions to be specified that exploit the beneficial properties of each representation.}

    \item For fixed network topologies, the framework recovers many previous decentralised stability conditions as special cases and allows conic combinations of them to be applied across different frequency ranges. Moreover, the full-block multipliers in the quadratic constraints offer additional degrees of freedom for constructing more advanced \lhh{and less restrictive} conditions.

    \item The framework accommodates detailed, heterogeneous, dynamic line models together with full-order bus models. As the constraints are specified in the frequency domain, they are compatible with black-box device models.

    \item Many of the quadratic constraints admit graphical interpretations, similar to \cite{lestas_NetworkStability_11, lestas_LargeScale_12, pates_StabilityCertificates_12,devane_PrimaryFrequency_17,huang_GainPhase_24,huang_GeometricDecentralized_25,leng_GeometricDecentralized_26}, that restrict the numerical range of the subsystem frequency response \lhg{and its higher-dimensional generalisation, via the Davis--Wielandt (DW) shell,} to a specified region, aiding the selection \lhg{of appropriate local specifications}.

\end{itemize}

\lhf{To facilitate readability}, this paper is split into two parts. Part I develops the stability framework under the assumption that either the impedance or admittance representation of the grid is stable.
\lhf{The framework is then demonstrated on a modified 9-bus case study, where we show that the proposed full-block conditions hold in frequency ranges where previously reported conditions fail.} 

\lhf{However, this case study highlights \lhg{some} limitations of the impedance and admittance representations which motivate the extensions presented in Part II. \lhg{In particular},} the open-loop bus impedance is unstable, \lhf{limiting the ability to utilise conditions that allow for plug-and-play operation. In Part II, we show that these \lhg{subsystem instabilities are fundamental features} arising in many common scenarios,}
\lhg{and that they may be stabilised using a hybrid impedance/power representation, which allows us to derive local conditions that enable plug-and-play operation in this more involved setting.}

The rest of Part I is structured as follows. \Cref{sec:background} presents the necessary mathematical background and Nyquist stability theory, and \cref{sec:model} introduces the grid model. \Cref{sec:stab} then details the small-signal stability framework: results are presented at three levels of locality in \crefrange{sec:stab_decentralised}{sec:stab_fixed} before being combined in \cref{sec:stab_comb}. The case study is presented in \cref{sec:case_study}, and \cref{sec:conclusion} concludes Part I.

\section{Mathematical Background} \label{sec:background}
\subsection{Notation and Definitions}
\lhe{
The sets $\mathbb{R}$, $\mathbb{C}$, $j\mathbb{R}$, $\mathbb{C}_+$,  $\bar{\mathbb{C}}_+$, and $\mathbb{H}^n$ denote the real numbers, complex numbers, imaginary axis, open right half-plane, closed right half-plane, and Hermitian matrices, respectively. $I_n$ and $0_n$ denote the identity and zero matrices, respectively (the dimension will be omitted when it is clear from the context). Define the matrix $J = \begin{bsmallmatrix} 0 & -1 \\ 1 & 0\end{bsmallmatrix}$. For matrices $A$, $B$, we let $A \otimes B$ denote the Kronecker product and $\sigma(A)$ denote the spectrum. For $A \in \mathbb{H}^n$, $A > 0$ ($A \geq 0$) denotes positive (semi)definiteness. We let $\sqrt{A}$ denote any matrix such that $(\sqrt{A})^2 = A$ (where it exists), and we define $\sqrt{A^\ast} := \left( \sqrt{A} \right)^\ast$. 
}
\lhe{
For an ordered indexed set $\mathcal{K}=\{\kappa_1, \ldots, \kappa_{|\mathcal{K}|}\}$, define $\oplus_{\kappa_k \in \mathcal{K}} B_k = \mathrm{diag} (B_1, \ldots, B_{|\mathcal{K}|})$ and $[a_k]_{\kappa_k \in \mathcal{K}} =[a_1^T, \ldots,a_{|\mathcal{K}|}^T]^T$ for matrices $B_k \in \mathbb{C}^{m_k \times n_k}$ and vectors $a_k \in \mathbb{C}^{n_k}$ associated with each $\kappa_k \in \mathcal{K}$.
}
\begin{definition} \label{def:numerical_range}
    The \textit{numerical range} or \textit{field of values} \cite{horn_TopicsMatrix_94} of a matrix $A \in \mathbb{C}^{n \times n}$ is defined as 
    \begin{equation*} 
        \mathcal{W}(A):= \{x^\ast A x \  | \ x \in \mathbb{C}^n, \lVert x \rVert = 1 \}.
    \end{equation*}
\end{definition}
We note that $\mathcal{W}(A)$ is a convex, compact subset of $\mathbb{C}$ which contains the spectrum of $A$ \cite{horn_TopicsMatrix_94}. We also have the following result \cite[Property 1.2.5]{horn_TopicsMatrix_94}. 
\begin{lemma} \label{lemma:numerical_range_pos_def}
    Let $A \in \mathbb{C}^{n\times n}$. Then $A+A^\ast$ is positive definite (resp. positive semidefinite) if and only if $\mathcal{W}(A) \subset \mathbb{C}_+$ (resp. $\mathcal{W}(A) \subset \bar{\mathbb{C}}_+$). 
\end{lemma}
\lhf{We also have the following convex three-dimensional generalisation of the numerical range \cite{lestas_LargeScale_12}.}
\begin{definition} \label{def:dw_shell}
    \lhf{The \textit{Davis--Wielandt (DW) shell} of a matrix $A \in \mathbb{C}^{n \times n}$ is defined as} 
    \begin{equation*} 
        \lhf{\mathcal{DW}(A):= \{(x^\ast A x,\lVert Ax\rVert^2) \  | \ x \in \mathbb{C}^n, \lVert x \rVert = 1 \}}.
    \end{equation*}
\end{definition}

\subsection{Nyquist Theory}
Let $\mathbf{R}^{m\times n}$ be the set of proper, \lhf{real}-rational transfer-function matrices of dimension $m \times n$.
Let $\mathbf{RH}_\infty^{m \times n} \subseteq \mathbf{R}^{m \times n}$ and \lhg{$\mathbf{RH}_{\infty,0}^{m \times n} \subseteq \mathbf{R}^{m\times n}$} be the sets of transfer functions whose elements have no poles in $\bar{\mathbb{C}}_+$ and \lhg{in $\bar{\mathbb{C}}_+ \setminus \{0\}$, respectively}.
\begin{definition} \label{def:positive_real}
    \lhe{A real-rational transfer-function matrix} $G(s)$ is \textit{positive-real} (\textit{resp.} \textit{strictly positive-real}) if \lhe{it has no poles in $\mathbb{C}_+$,} $G(s) + G(s)^\ast \geq 0$ (\textit{resp.} $G(s) + G(s)^\ast > 0$) for all $s \in \mathbb{C}_+$ \lhe{and at all points of analyticity on $j\mathbb{R}$, and any imaginary-axis poles are simple with positive-semidefinite residue matrices~\cite{khalil_NonlinearSystems_13}}. 
\end{definition}

Let $P(s)$ and $K(s)$ be two transfer-function matrices \lhe{of dimension $n \times m$ and $m \times n$, respectively, where $L(s):=P(s)K(s) \in \mathbf{R}^{n\times n}$} and let $[P(s),K(s)]$ denote their negative-feedback interconnection. We say $[P(s),K(s)]$ is \textit{well-posed} if all closed-loop transfer matrices are well defined and proper \cite{zhou_RobustOptimal_96}. For the purposes of this paper, we use the following definition for stability.
\begin{definition} \label{def:stability}
    We say the interconnection is \textit{stable} if $(I+P(s)K(s))^{-1} \in \mathbf{RH}_{\infty,0}^{n\times n}$ \lhh{and, at $s=0$, has} \lhg{at most one pole.}
\end{definition}

\lhg{Poles are counted with multiplicity, so this definition allows the closed-loop system to have a single pole} at the origin, \lha{which is common in power systems stability studies \cite[Section 12.7]{kundur_PowerSystem_22}}. 
This also implies internal stability \lhe{(in the sense \lhg{of \cref{def:stability}})} if there are no pole cancellations in $\bar{\mathbb{C}}_+$ in the product $P(s)K(s)$.

The generalised Nyquist stability theorem \cite{maciejowski_MultivariableFeedback_89, desoer_GeneralizedNyquist_80} gives a necessary and sufficient condition for the feedback configuration $[P(s),K(s)]$ to be stable. We define the \textit{Nyquist contour} \lhe{with an indentation around the origin}, illustrated in \cref{fig:tikz_nyquist_mod}, as the closed curve on the complex plane consisting of the imaginary axis \lhe{with the segment $(-j\epsilon,j\epsilon)$ replaced by a semicircular indentation into $\mathbb{C}_+$}, and an infinite\footnote{The limit $R\to\infty$ in \eqref{eq:nyquist_contour_mod} is considered in the compactified complex plane.} semicircle in the right half-plane joining the positive imaginary axis to the negative imaginary axis \cite{skogestad_MultivariableFeedback_05}:
\begin{equation} \label{eq:nyquist_contour_mod}
    \Gamma_N^{\,\epsilon} = \lim_{R \rightarrow \infty}
    \left\{ \Gamma_{j\omega}^{\epsilon-} \cup \Gamma_\epsilon \cup \Gamma_{j\omega}^{\epsilon+} \cup \Gamma_R \right\},
\end{equation}
where
    $ \Gamma_{j\omega}^{\epsilon-} = \left\{ s = j\omega \ | \ \omega \in \left[-R,-\epsilon \right] \right\}$,
    $ \Gamma_\epsilon = \left\{s =  \epsilon e^{j \varphi} \ | \ \varphi \in \left[ -\frac{\pi}{2}, \frac{\pi}{2} \right] \right\}$,
    $ \Gamma_{j\omega}^{\epsilon+} = \left\{ s = j\omega \ | \ \omega \in \left[\epsilon,R \right] \right\}$, 
    $\Gamma_R = \left\{s = Re^{j\varphi} \ | \ \varphi \in [\frac{\pi}{2}, -\frac{\pi}{2}]\right\}$,
and $\epsilon>0$ is sufficiently small.

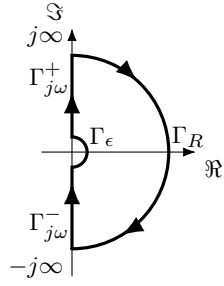
\begin{figure}[t]
\centering
\input{Figures/tikz_nyquist_mod}
\caption{Illustration of the indented Nyquist contour.}
\label{fig:tikz_nyquist_mod}
\end{figure}

\begin{theorem} \label{thm:nyquist}
    Assume the closed-loop system $[P(s),K(s)]$ is well-posed. Let $L(s) = P(s)K(s) \in \mathbf{R}^{n\times n}$ be the open-loop transfer function and let $L(s)$ have $P_{ol}$ poles in $\mathbb{C}_+$ \lhe{and} no poles on \lhg{$j\mathbb{R} \setminus \{0\}$}. Let $N_{\Gamma}$ denote the net number of clockwise encirclements made by the plots of $\sigma(L(s))$ around the point $-1$ as $s$ traverses $\Gamma_N^{\, \epsilon}$ with $\epsilon\to0$.
    Then,
    \begin{equation} \label{eq:nyquist_pole_count}
        N_\Gamma = P_{cl} - P_{ol},
    \end{equation}
    where $P_{cl}$ gives the number of poles in $\mathbb{C}_+$ of the closed-loop transfer function $(I+P(s)K(s))^{-1}$.
\end{theorem}

Now consider the case where $L(s) \in \mathbf{RH}_{\infty,0}^{n \times n}$, i.e., \lhe{the open-loop transfer function has no \lhg{poles in $\bar{\mathbb{C}}_+ \setminus \{0\}$}}, and where $[P(s),K(s)]$ is well-posed. Then \cref{thm:nyquist} gives that $[P(s),K(s)]$ is stable if and only if, \lhe{in addition to $(I+L(s))^{-1}$ having at most \lhg{one} pole at the origin,} the plots of $\sigma(L(s))$ as $s$ traverses the indented Nyquist contour (the \textit{characteristic loci} of $L(s)$) make no net encirclements of the point $-1$, and do not pass through the point $-1$. The following lemma provides a sufficient condition for this to hold. 

\begin{lemma} \label{lemma:nyquist_sufficient}
    Assume $L(s) = P(s)K(s) \in \mathbf{RH}_{\infty,0}^{n\times n}$, the closed-loop system $[P(s),K(s)]$ is well-posed and, \lhh{at $s=0$,} $(I + L(s))^{-1}$ has at most \lhg{one} pole. 
    \lhf{Let $\sigma(L_\tau(s))$ be a continuous deformation of the characteristic loci $\sigma(L(s))$ parametrised by $\tau \in [0,1]$ such that $\sigma(L_0(s)) = \{0\}$ and $\sigma(L_1(s)) = \sigma(L(s))$.}
    Then, $[P(s),K(s)]$ is stable under \cref{def:stability} if there exists an $\bar{\epsilon}>0$ such that for all $\epsilon \in (0,\bar{\epsilon}]$,
    \begin{equation} \label{eq:nyquist_condition}
        -1 \notin \sigma(L_\tau(s)), \quad \forall s \in \Gamma_{N}^{\, \epsilon}, \ \forall \tau \in [0,1].
    \end{equation}
\end{lemma}
\begin{proof}
    The lemma is proved in Appendix \ref{app:nyquist_sufficient_proof}.
\end{proof}

We now state a sufficient condition that ensures \eqref{eq:nyquist_condition} is satisfied \lhf{for a linear homotopy $L_\tau(s) = \tau L(s), \tau \in [0,1]$}. While the proof here uses a Nyquist argument (similar to \cite{laib_DecentralizedStability_23} and \cite{ringh_GainTypePhaseType_26}), \lhg{an analogous stability result} can be obtained by applying the theory of integral quadratic constraints (IQCs) \cite{megretski_SystemAnalysis_97} or graph-separation methods \cite{iwasaki_WellposednessFeedback_98} to the operators $P(s)$ and $K(s)$.

\begin{lemma} \label{lemma:iqc}
    \lhf{Let $L_\tau(s) = \tau L(s)$ for $\tau \in [0,1]$.} Then, condition \eqref{eq:nyquist_condition} for the feedback system $[P(s),K(s)]$ holds \lhe{at a particular $s \in \Gamma_N^{\, \epsilon}$ \lhh{and $\tau$}} if there exists a matrix (\lhe{referred to as} a \textit{multiplier})
    \begin{equation} \label{eq:multiplier}
        \Pi(s) = \begin{bmatrix}
            \Pi_{11}(s) & \Pi_{12}(s)^\ast \\ \Pi_{12}(s) & \Pi_{22}(s)
        \end{bmatrix},
    \end{equation}
    where $\Pi_{11}(s) \in \mathbb{H}^n$, $\Pi_{12}(s) \in \mathbb{C}^{m \times n}$ and $\Pi_{22}(s) \in \mathbb{H}^m$, such that the following conditions are satisfied: 
    \begin{subequations} \label{eq:iqc_conditions}
         \begin{align}
             \begin{bmatrix}
                P(s) \\ I_m
            \end{bmatrix}^\ast &\Pi(s)
            \begin{bmatrix}
                P(s) \\ I_m
            \end{bmatrix} > 0, \label{eq:iqc_conditions_top}
            \\
             \begin{bmatrix}
                I_n \\ -\tau K(s)
            \end{bmatrix}^\ast &\Pi(s)
            \begin{bmatrix}
                I_n \\ - \tau K(s)
            \end{bmatrix} \leq 0. \label{eq:iqc_conditions_btm}
        \end{align}
    \end{subequations} 
\end{lemma}
\begin{proof}
    The lemma is proved in Appendix \ref{app:iqc_proof}.
\end{proof}
The conditions \eqref{eq:iqc_conditions} therefore define sufficient quadratic constraints on the input--output properties of the systems $P(s)$ and $K(s)$ 
\lhe{under which} the characteristic loci of
\lhe{$\tau L(s)$ avoid $-1$ at a given $s \in \Gamma_{N}^{\, \epsilon}$ \lhh{and $\tau \in [0,1]$}.}
\lhe{If such multipliers exist at every $s \in \Gamma_{N}^{\, \epsilon}$ \lhh{and every $\tau \in [0,1]$} (for all $\epsilon \in (0,\bar\epsilon]$), then $[P(s),K(s)]$ is stable by \cref{lemma:nyquist_sufficient}. However, \eqref{eq:iqc_conditions_btm} requires a search over $\tau \in [0,1]$; we now give a multiplier structure under which it need only be checked at $\tau = 1$.
}
\begin{corollary} \label{cor:iqc_convex}
    Let $\Pi_{11}(s) \leq 0$ and $\Pi_{22}(s) \geq 0$ in the multiplier $\Pi(s)$ given in \eqref{eq:multiplier} and let \eqref{eq:iqc_conditions_btm} hold for $\tau = 1$. Then, \eqref{eq:iqc_conditions_btm} holds for all $\tau \in [0,1]$.  
\end{corollary}
\begin{proof}
    The left-hand side of \eqref{eq:iqc_conditions_btm} expands to
    \begin{equation} \label{eq:prf_iqc_eqiv_condition}
        \begin{aligned}
            \Pi_{11}&(s) - \tau \Pi_{12}(s)^\ast K(s) \\ &- \tau K(s)^\ast \Pi_{12}(s) + \tau^2 K(s)^\ast \Pi_{22}(s) K(s).
        \end{aligned}
    \end{equation}
    As $\Pi_{22}(s) \geq 0$, \eqref{eq:prf_iqc_eqiv_condition} is convex in $\tau$. For $\tau = 0$, \eqref{eq:prf_iqc_eqiv_condition} is negative semidefinite (as $\Pi_{11}(s) \leq 0$), and for $\tau = 1$, \eqref{eq:prf_iqc_eqiv_condition} is negative semidefinite by assumption. Therefore, \eqref{eq:iqc_conditions_btm} holds for all $\tau \in [0,1]$. 
\end{proof}

\begin{remark}
    Similar results can be obtained using IQC-based methods \cite{megretski_SystemAnalysis_97}. However, the approach taken in \cref{lemma:iqc} can be extended to improper systems, which often appear in power systems, while classical IQC results apply only to bounded systems. 
    \lhe{Furthermore, the homotopy in \eqref{eq:nyquist_condition} can be replaced by any continuous deformation of the loci, allowing alternative system representations with homotopies in different geometries. \lhg{It will be shown in \cref{sec:stab} that} these reconfigurations and multipliers can be mixed over different frequency ranges.}
\end{remark}

\section{Grid Model} \label{sec:model}
\subsection{Topology and Reference Frames} \label{sec:model_topology}
The three-phase AC grid is composed of $N_B$ buses, forming the node set $\mathcal{V}_B := \{\nu_1, \ldots, \nu_{N_B}\}$, and $N_P$ power lines, assigned to the edge set $\mathcal{E}_P \subseteq \mathcal{V}_B \times \mathcal{V}_B$, with $(\nu_i,\nu_j) \in \mathcal{E}_P$ indicating a line from $\nu_j$ (source node) to $\nu_i$ (sink node).
In addition, we include a node representing ground, denoted $\nu_0$. Loads and other shunt connections link elements of $\mathcal{V}_B$ to $\nu_0$, and are assigned to the set $\mathcal{E}_0 \subseteq \mathcal{V}_B \times \{\nu_0\}$.
Together, the power grid is represented by the directed graph $(\mathcal{V}_N,\mathcal{E}_N)$, where $\mathcal{V}_N = \mathcal{V}_B \cup\{\nu_0\}$ and $\mathcal{E}_N = \mathcal{E}_P \cup \mathcal{E}_0 $.
\begin{assumption} \label{assump:connected}
    The subgraph $(\mathcal{V}_B,\mathcal{E}_P) \subseteq (\mathcal{V}_N,\mathcal{E}_N)$ is connected.
\end{assumption}

For each edge $(\nu_i,\nu_j) \in \mathcal{E}_P$ (i.e., the power lines only), assign an index $k$ so that $\varepsilon_k \equiv (\nu_i,\nu_j)$ and order the set so that $\mathcal{E}_P = \{\varepsilon_1,\ldots,\varepsilon_{N_P}\}$. The $(i,k)^{\mathrm{th}}$ entry of the incidence matrix $B_N\in\mathbb{R}^{N_B \times N_P}$ for the bus--line subgraph $(\mathcal{V}_B,\mathcal{E}_P)$ is given by
\begin{equation} \label{eq:incidence_matrix}
    B_N^{(i,k)} = \begin{cases}
        1, & \text{if $\varepsilon_k \equiv (\nu_i,\nu_j) \in \mathcal{E}_P$, for some $\nu_j$}, \\
        -1, & \text{if $\varepsilon_k \equiv (\nu_j,\nu_i) \in \mathcal{E}_P$, for some $\nu_j$}, \\
        0, & \text{otherwise.}
    \end{cases}
\end{equation}

At each bus $\nu_i \in \mathcal{V}_N$, we associate a balanced three-phase voltage $v_i(t) \in \mathbb{R}^3$ and net current injection $i_i(t) \in \mathbb{R}^3$.
Since all signals are balanced, they can be projected into a set of synchronously rotating $DQ$ reference frames
\cite{kundur_PowerSystem_22, schiffer_SurveyModeling_16, spanias_SystemReference_19}.

We first define a \textit{common} $DQ$ reference frame rotating at constant angular frequency $\omega_0 >0$.
A balanced three-phase signal $x(t) \in \mathbb{R}^3$ evaluated in this frame is denoted $x^{DQ}(t)=[x^D(t),x^Q(t)]^T \in \mathbb{R}^2$.
Similarly, with each bus $\nu_i \in \mathcal{V}_B$, we associate a \textit{local} $dq$ reference frame rotating at the dynamic local angular frequency $\omega_i(t)>0$,
and the same signal
\lhe{in this frame is denoted} $x^{dq}(t)=[x^d(t),x^q(t)]^T \in \mathbb{R}^2$.
The angular difference between the local 
and common reference frames \lhe{satisfies}
\begin{equation} \label{eq:delta_dot}
    \dot{\delta}_i(t) = \omega_i(t) - \omega_0,
\end{equation}
\lhf{which implies that $\omega_i(t) = \omega_0$ at equilibrium.}
Therefore, we can map between the reference frames as
$ 
    x^{DQ}(t) = T(\delta_i(t))x^{dq}(t),
$ 
where
\begin{equation} \label{eq:dq_transform}
    T(\delta_i(t)) = \begin{bmatrix}
        \cos(\delta_i(t)) & -\sin(\delta_i(t)) \\ \sin(\delta_i(t)) & \cos(\delta_i(t))
    \end{bmatrix}.
\end{equation}
Without loss of generality, we define the local $d$-axis to be aligned with the bus voltage vector so that $v_i^{dq}(t) = [V_i(t), 0]^T$, where $V_i(t) > 0$ is the voltage magnitude. 
Unless it is otherwise ambiguous, we henceforth omit the argument $t$ from time-domain signals.

\subsection{Buses} \label{sec:model_buses}
At each bus $\nu_i \in \mathcal{V}_N$, we consider a dynamical system whose output is the bus voltage $v_i^{DQ} \in \mathbb{R}^2$ and whose input is the negative net current injection $-i_i^{DQ} \in \mathbb{R}^2$, both evaluated in the common reference frame. Kirchhoff's current law gives
\begin{equation} \label{eq:kcl}
    i_i^{DQ} = \textstyle \sum_{\nu_j \in \mathcal{N}_i^+} i_{ij}^{DQ} - \sum_{\nu_j \in \mathcal{N}_i^-} i_{ji}^{DQ},
\end{equation}
where $i_{ij}^{DQ} \in \mathbb{R}^2$ is the current through the branch $(\nu_i,\nu_j) \in \mathcal{E}_N$ flowing from $\nu_i$ to $\nu_j$, \lhe{and $\mathcal{N}_i^+$ and $\mathcal{N}_i^-$ are the sets of in- and out-neighbours of $\nu_i \in \mathcal{V}_B$ in $\mathcal{V}_N$, respectively}.
Note that a double index on $i_{ij}^{DQ}$ denotes a branch current associated with the line/load $(\nu_i,\nu_j) \in \mathcal{E}_N$, while a single index on $i_i^{DQ}$ denotes a net nodal current injection associated with $\nu_i \in \mathcal{V}_N$.

The relationship between $-i_i^{DQ}$ and $v_i^{DQ}$ is determined by the dynamics of the connected device, such as a synchronous generator or a grid-forming inverter. As we are interested in the small-signal stability of the grid, we linearise each system, \lhe{together with the frame transformation,} about \lhe{the network} equilibrium
\lhe{to get} the $2 \times 2$ bus impedance $Z_i(s)$
\lhe{in the} Laplace domain,
\lhe{defined by}
\begin{equation} \label{eq:bus_iv_relationship}
    \Delta v_i^{DQ}(s) = - Z_i(s) \Delta i_i^{DQ}(s).
\end{equation}
Here, $\Delta v_i^{DQ}(s)$ and $\Delta i_i^{DQ}(s)$ are the Laplace-domain deviations of $v_i^{DQ}$ and $i_i^{DQ}$ about their operating points, respectively. As a special case, we model the ground node $\nu_0$ as a constant voltage source with zero output, meaning $\Delta v_0^{DQ}(s) = [0, 0]^T$.

\subsection{Network} \label{sec:model_network}
We associate with each line or load $(\nu_i,\nu_j) \in \mathcal{E}_N$ a subsystem whose input is the potential difference across the element in the common $DQ$ reference frame, given by the difference in the voltages of the adjacent buses $v_i^{DQ} - v_j^{DQ}$, and whose output is the associated branch current, $i_{ij}^{DQ}$. In the Laplace domain, we obtain
\begin{equation} \label{eq:line_dynamics_s}
    \Delta i_{ij}^{DQ}(s) = Y_{ij}(s) \left( \Delta v_i^{DQ}(s) - \Delta v_j^{DQ}(s) \right),
\end{equation}
where $Y_{ij}(s)$ is the branch admittance for the line $(\nu_i,\nu_j)\in\mathcal{E}_N$. For lines \lhf{and constant-impedance} loads represented by series $RL$ components\footnote{\lhb{We note that other line/load models can be used. In such cases, $Y_{ij}(s)$ is the transfer function obtained by linearising the branch model about the system equilibrium.}}, we have
\begin{equation} \label{eq:line_admittance}
    Y_{ij}(s) = \begin{bmatrix}
        R_{ij}+X_{ij}\frac{s}{\omega_0} & - X_{ij} \\ X_{ij} & R_{ij}+X_{ij}\frac{s}{\omega_0}
    \end{bmatrix}^{-1},
\end{equation}
where $R_{ij} \geq  0$ is the line/load resistance and $X_{ij} \geq 0$ is the line/load reactance, both measured in p.u. \lhe{For a shunt capacitive branch, we have
\begin{equation} \label{eq:shunt_admittance}
    Y_{i0}(s) = \begin{bmatrix}
        B_{i0}\frac{s}{\omega_0} & -B_{i0} \\ B_{i0} & B_{i0}\frac{s}{\omega_0}
    \end{bmatrix},
\end{equation}
where $B_{i0} \geq 0$ is the shunt susceptance in p.u.}

For a load or shunt $(\nu_i,\nu_0) \in \mathcal{E}_0$, $\Delta v_0^{DQ}(s) = [0,0]^T$, and we write
\begin{equation} \label{eq:shunt_admittance_simp}
    Y_i(s) \equiv Y_{i0}(s),
\end{equation}
setting $Y_{i0}(s)$ to $0_2$ if no load/shunt is present.

Now, let $\Delta v_B^{DQ}(s) = [\Delta v_i^{DQ}(s)]_{\nu_i \in \mathcal{V}_B}$ and $\Delta i_B^{DQ}(s) = [\Delta i_i^{DQ}(s)]_{\nu_i \in \mathcal{V}_B}$. Using \eqref{eq:line_dynamics_s} in Kirchhoff's current law \eqref{eq:kcl} at each node $\nu_i \in \mathcal{V}_B$ and using the fact that $Y_{ji}(s) \equiv Y_{ij}(s)$ (by symmetry in the line model) allows us to write
\begin{equation} \label{eq:kcl_matrix_form}
    \Delta i_B^{DQ}(s) = Y_N(s) \Delta v_B^{DQ}(s),
\end{equation}
where $Y_N(s)$ is the network admittance matrix. The matrix $Y_N(s)$ is composed of $2 \times 2$ blocks, where the $(i,j)^{\mathrm{th}}$ block is
\begin{equation} \label{eq:network_admittance_blocks}
    Y_N^{(i,j)}(s) = \begin{cases}
          Y_i(s) + \sum\limits_{\nu_k \in \mathcal{N}_i}  Y_{ik}(s), & i=j \\
        -Y_{ij}(s), &  \lhh{\nu_j \in \mathcal{N}_i} \\
        0_2, & \text{otherwise,}
    \end{cases}
\end{equation}
where $Y_i(s)$ is given by \eqref{eq:shunt_admittance_simp} \lhe{and $\mathcal{N}_i$ is the set of buses in $\mathcal{V}_B$ neighbouring $\nu_i \in \mathcal{V}_B$}. Equivalently,
\begin{equation} \label{eq:network_admittance}
    Y_N(s) = \mathcal{B}_N Y_P(s) \mathcal{B}_N^T + Y_0(s),
\end{equation}
where $\mathcal{B}_N = B_N \otimes I_2$ (with $B_N$ defined in \eqref{eq:incidence_matrix}), $Y_P(s) = \oplus_{(\nu_i,\nu_j) \in \mathcal{E}_P} Y_{ij}(s)$ (with block order following that used for $B_N$ in \eqref{eq:incidence_matrix}), and $Y_0(s) = \oplus_{\nu_i \in \mathcal{V}_B} Y_i(s)$, with $Y_i(s)$ as in \eqref{eq:shunt_admittance_simp}.

\subsubsection*{Properties} 
It is well known that $Y_{ij}(s)$ in \eqref{eq:line_admittance} and \eqref{eq:shunt_admittance} is positive-real (see \cref{def:positive_real}) \lhe{and strictly positive-real for $R_{ij}>0$ in \eqref{eq:line_admittance}}. Therefore, $Y_N(s)$ \lhf{composed of blocks modelled as \eqref{eq:line_admittance} and \eqref{eq:shunt_admittance}}, is positive-real \cite{watson_ScalableControl_21,chen_ExtendedFrequencyDomain_24,hallinan_PerformanceStability_25} and \lhf{has no poles in $\mathbb{C}_+$}. \lhf{This property will be exploited in \cref{sec:stab_decentralised}. Note that the inclusion of loads not modelled as \eqref{eq:line_admittance} may make $Y_N(s)$ non-positive-real; however, in such cases, the dynamics of a non-positive-real load at bus $\nu_i \in \mathcal{V}_B$ may be included in the bus impedance $Z_i(s)$, meaning the positive-real property of $Y_N(s)$ can still be utilised}. 

We also have the following property of $Y_{ij}(s)$ in \eqref{eq:line_admittance} \cite{vorobev_DecentralizedStability_19,chen_ExtendedFrequencyDomain_24}. 
\begin{lemma} \label{lemma:rotated_admittance}
    Consider $Y_{ij}(s)$ as given in \eqref{eq:line_admittance}, with $X_{ij} \geq 0$. For $|\omega| \leq \omega_0$, $JY_{ij}(j\omega) - Y_{ij}(j\omega)^\ast J\geq 0$, where $J = \begin{bsmallmatrix}
        0 & -1 \\ 1 & 0 
    \end{bsmallmatrix}$. 
\end{lemma}
\begin{proof}
    By direct calculation, $\frac{1}{2}({Y_{ij}(j\omega)^{\ast}}^{-1}J - JY_{ij}(j\omega)^{-1}) = X_{ij}I_2 - jX_{ij}\frac{\omega}{\omega_0}J$, 
    which has eigenvalues $X_{ij} \pm X_{ij}\frac{\omega}{\omega_0}$. Therefore, ${Y_{ij}(j\omega)^{\ast}}^{-1}J - JY_{ij}(j\omega)^{-1} \geq 0$ for $|\omega| \leq \omega_0$ and $X_{ij} \geq 0$. By post-multiplying by $Y_{ij}(j\omega)$ and pre-multiplying by $Y_{ij}(j\omega)^\ast$, we arrive at the result. 
\end{proof}

\lhe{While this property holds for branches modelled by \eqref{eq:line_admittance}, it does not hold for the capacitive shunts \eqref{eq:shunt_admittance}, which appear as blocks of $Y_0(s)$ in \eqref{eq:network_admittance}.}
Therefore, the sign definiteness of the matrix $(I_{N_B} \otimes J)Y_N(j\omega) - Y_N(j\omega)^\ast (I_{N_B} \otimes J)$ will depend on the exact network configuration. 
\lhf{However, if there exists an $\omega_r$ such that this property is known to hold for all $|\omega| \leq \omega_r$, it can be exploited, as shown in \cref{sec:stab_fixed}.}

\subsection{Closed-Loop Grid} \label{sec:model_closed_loop}
Let
    $ Z_B(s) = \oplus_{\nu_i \in \mathcal{V}_B} Z_i(s) $.
Then \eqref{eq:bus_iv_relationship} can be written in matrix form as
\begin{equation} \label{eq:kvl_matrix_form}
    \Delta v_B^{DQ} = - Z_B(s) \Delta i_B^{DQ}.
\end{equation}
Together, \eqref{eq:kcl_matrix_form} and \eqref{eq:kvl_matrix_form} form the negative-feedback interconnection between the bus dynamics and the network admittance matrix $[Z_B(s),Y_N(s)]$, as illustrated in \cref{fig:tikz_model_iv}.

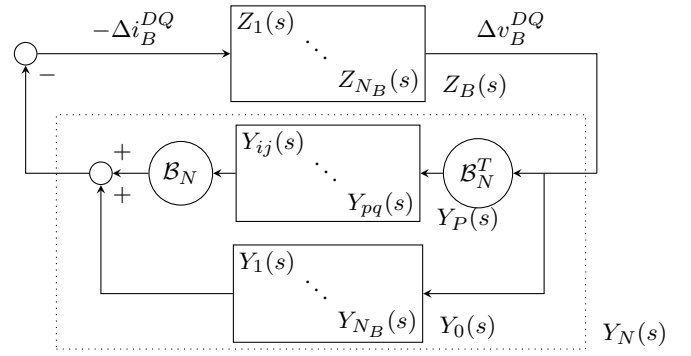
\begin{figure}[t]
\centering
\input{Figures/tikz_model_iv}
\caption{Negative-feedback interconnection of the bus impedances in $Z_B(s)$ and the network admittance $Y_N(s)$.}
\label{fig:tikz_model_iv}
\end{figure}

\lhc{The small-signal stability of the grid can therefore be concluded by analysing the feedback system $[Z_B(s),Y_N(s)]$, \lhe{which we refer to as the \textit{impedance
representation} of the grid}. Alternatively, by inverting both systems, \lhe{we obtain the \textit{admittance representation}} $[\widehat{Y}_B(s), \widehat{Z}_N(s)]$, where}%
\begin{equation} \label{eq:inverse_systems}
    \widehat{Y}_B(s) = Z_B(s)^{-1}, \quad \widehat{Z}_N(s) = Y_N(s)^{-1}.
\end{equation}
\lhe{We note that $\widehat{Z}_N(s)$ exists provided $Y_0(s) \not\equiv 0$, i.e., at least one load or shunt is present, \lhf{which typically holds in practice} (otherwise $Y_N(s) = \mathcal{B}_N Y_P(s) \mathcal{B}_N^T$ is identically singular by \cref{assump:connected}). Furthermore, we note that \lhf{if $Y_N(s)$ is positive-real}, $\widehat{Z}_N(s)$ has no poles in $\mathbb{C}_+$ \cite{hallinan_PerformanceStability_25}.}

We note that \lhf{without a stiff voltage source in the network, there is no} absolute angle reference when defining $\omega_0$ in \eqref{eq:delta_dot}. \lhf{Therefore,} the feedback systems $[Z_B(s), Y_N(s)]$ \lha{and $[\widehat{Y}_B(s),\widehat{Z}_N(s)]$} \lhf{may} exhibit a pole at the origin \lhf{(\lhg{there is a single such pole} provided frequency control or damping is included in the bus model)} \cite[Section 12.7]{kundur_PowerSystem_22}. To conclude stability, we apply the definition presented in \cref{def:stability}. 
\lhf{In Part I of this paper, we use the following assumption, which is frequently used in the literature, so that \cref{lemma:nyquist_sufficient,lemma:iqc} can be directly applied.}
\begin{assumption} \label{assump:open_loop}
\lhc{Either\footnote{
\lhe{Lossless branches with $R_{ij}=0$ in \eqref{eq:line_admittance} may introduce imaginary-axis poles \lhf{at $s = \pm j\omega_0$} into $Y_N(s)$; in such cases, the results of \cref{sec:background} can be extended by adding semicircular indentations into $\mathbb{C}_+$ to \eqref{eq:nyquist_contour_mod} at each imaginary-axis pole. Capacitive shunts \eqref{eq:shunt_admittance} render $Y_N(s)$ improper; this is accommodated since only the open-loop products in \cref{assump:open_loop} are required to be proper.}}:
    \begin{enumerate}[A.]
        \item $L_N(s) = Z_B(s)Y_N(s) \in \mathbf{RH}_{\infty}^{2N_B \times 2N_B}$ and, \lhh{at $s=0$,} the system $(I + Z_B(s)Y_N(s))^{-1}$ has at most \lhg{one} pole; or
        \item $Y_N(s)$ is nonsingular, $\widehat{L}_N(s) = \widehat{Y}_B(s)\widehat{Z}_N(s) \in \mathbf{RH}_{\infty}^{2N_B \times 2N_B}$ and, \lhh{at $s=0$,} the system $(I + \widehat{Y}_B(s)\widehat{Z}_N(s))^{-1}$ has at most \lhg{one} pole. 
    \end{enumerate}
}
\end{assumption}

\lhc{As we will see, both the impedance and admittance representations have \lhf{complementary} advantages \lhf{and disadvantages} for stability analysis:}
\begin{itemize}
    \item \lhf{As we show in \cref{sec:stab_extended},} \lhe{the impedance representation} allows us to exploit the sparsity of $Y_N(s)$ to generate decentralised stability conditions \lhf{on extended subsystems containing buses and local lines and loads. These conditions can be designed to allow for plug-and-play operation}. However, as we show in Part II of this paper, $Z_B(s)$ contains unstable poles in many common scenarios, \lhf{which violates \cref{assump:open_loop}.A}.

    \item \lhf{Alternatively}, if $\widehat{Y}_B(s)$ is stable, \lhe{the admittance representation can be analysed instead}. However, the inverted $\widehat{Z}_N(s)$ is in general non-sparse and introduces non-local coupling, meaning \lhf{plug-and-play compatible conditions on} extended subsystems composed of local dynamics cannot easily be formed. 
\end{itemize} 

\lhf{\Cref{assump:open_loop} is used here to facilitate readability \lhg{when formulating} the stability framework. In Part II of this paper, we relax \cref{assump:open_loop}.A so that plug-and-play compatible stability conditions can be derived on extended subsystems for networks with unstable $Z_B(s)$.}

\section{Stability Framework} \label{sec:stab}
We now introduce a stability framework
\lhe{for certifying that the grid}
is small-signal stable. 
\lhc{Throughout this section, we present the results using the
\lhe{impedance representation,} i.e., the open-loop transfer function $L_N(s) = Z_B(s) Y_N(s)$, and discuss in \cref{sec:stab_comb} how they apply to
\lhe{the admittance representation with} $\widehat{L}_N(s) = \widehat{Y}_B(s) \widehat{Z}_N(s)$}. 

\lhe{Under \cref{assump:open_loop}.A, \cref{lemma:nyquist_sufficient} guarantees that $[Z_B(s),Y_N(s)]$ is stable if, at every $s \in \Gamma_N^{\,\epsilon}$
(for all $\epsilon \in (0,\bar{\epsilon}]$),}
\begin{equation} \label{eq:loci_condition}
    -1 \notin \sigma(\tau L_N(s)), \quad \forall \tau \in [0,1].
\end{equation}
The framework consists of a set of quadratic constraints on the input--output properties of local grid subsystems, 
\lhe{each of which is sufficient for \eqref{eq:loci_condition} to hold at a given $s \in \Gamma_N^{\,\epsilon}$. Different conditions can be applied over different regions of the Nyquist contour, and stability is certified when these regions collectively cover the entire contour.}

To facilitate the analysis, we classify the stability conditions according to the level of locality of the subsystem to which they apply:
\begin{enumerate}
    \item \textbf{Decentralised bus conditions (\cref{sec:stab_decentralised}):} can be verified independently at each bus in the network and require \lhb{no global network information}.
    \item \textbf{Decentralised extended-subsystem conditions (\cref{sec:stab_extended}):} involve subsystems that couple each bus to its neighbouring lines along with any local loads.
    \item \lhb{\textbf{Local bus conditions for fixed networks (\cref{sec:stab_fixed}):} conditions (previously reported in the literature) that depend on centralised network information.}
\end{enumerate}

\lhf{The results in \crefrange{sec:stab_decentralised}{sec:stab_fixed} are presented as pointwise conditions over the Nyquist contour $\Gamma_N^{\, \epsilon}$.}
In \cref{sec:stab_comb}, we
\lhe{combine conditions from \lhg{each} level of locality into formal
stability certificates} \lhf{over the frequency domain}
and discuss how the framework can be practically implemented to design local grid codes that enable plug-and-play functionality. 

\subsection{Decentralised Bus Conditions} \label{sec:stab_decentralised}
We first apply \cref{lemma:iqc} to yield quadratic constraints on the local bus dynamics whenever a corresponding constraint on the network dynamics $Y_N(s)$ holds. The result exploits the block-diagonal structure of $Z_B(s)$. 

\begin{lemma} \label{lemma:iqc_decentralised}
    Consider the grid system $[Z_B(s),Y_N(s)]$ \lhe{described in \cref{sec:model} and} illustrated in \cref{fig:tikz_model_iv}. Let $\Pi(s) \in \mathbb{C}^{4N_B \times 4N_B}$ be a multiplier with the following block structure:
    \begin{equation} \label{eq:multiplier_decentralised}
        \Pi(s) = \begin{bmatrix}
            \oplus_{\nu_i \in \mathcal{V}_B} \Pi_{11}^i(s) & \oplus_{\nu_i \in \mathcal{V}_B} \Pi_{12}^i(s)^\ast \\
            \oplus_{\nu_i \in \mathcal{V}_B} \Pi_{12}^i(s) & \oplus_{\nu_i \in \mathcal{V}_B} \Pi_{22}^i(s)
        \end{bmatrix},
    \end{equation}
    where for each $\nu_i \in \mathcal{V}_B$, $\Pi_{11}^i(s) \in \mathbb{H}^2$, $\Pi_{12}^i(s) \in \mathbb{C}^{2\times 2}$, and $\Pi_{22}^i(s)\in \mathbb{H}^2$. 
    Then, for a particular $s \in \Gamma_N^{\,\epsilon}$, condition \eqref{eq:loci_condition} holds if, for every $\tau \in [0,1]$, there exists a $\Pi(s)$ of the form \eqref{eq:multiplier_decentralised} such that 
    \begin{subequations}    
    \begin{equation} \label{eq:iqc_decentralised_network}
        \begin{bmatrix} 
            I_{2N_B} \\ -\tau Y_N(s)
        \end{bmatrix}^\ast \Pi(s) \begin{bmatrix}
            I_{2N_B} \\ -\tau Y_N(s)
        \end{bmatrix} \leq 0,
    \end{equation}
    \begin{equation} \label{eq:iqc_decentralised_buses}
        \begin{bmatrix}
            Z_i(s) \\ I_2
        \end{bmatrix}^\ast \begin{bmatrix}
            \Pi_{11}^i(s) & \ \Pi_{12}^i(s)^\ast \\
            \Pi_{12}^i(s) &  \Pi_{22}^i(s)
        \end{bmatrix} \begin{bmatrix}
            Z_i(s) \\ I_2
        \end{bmatrix} > 0,
    \end{equation}
    \end{subequations}
    \lhh{where \eqref{eq:iqc_decentralised_buses} holds at every $\nu_i \in \mathcal{V}_B$.}
\end{lemma}
\begin{proof}
    This follows directly from \cref{lemma:iqc} with the block-diagonal structure of $Z_B(s)$ and the multiplier \eqref{eq:multiplier_decentralised}. 
\end{proof}

Consequently, if a multiplier $\Pi(s)$ of the form \eqref{eq:multiplier_decentralised} satisfies the network constraint \eqref{eq:iqc_decentralised_network}, then it is sufficient to verify the corresponding set of local conditions \eqref{eq:iqc_decentralised_buses} at each bus $\nu_i \in \mathcal{V}_B$ to ensure \eqref{eq:loci_condition} holds. Furthermore, if the multiplier $\Pi(s)$ is independent of the dynamics of other buses and network-level parameters, then the condition \eqref{eq:iqc_decentralised_buses} is fully decentralised and can be used for plug-and-play functionality. 

\lhf{The following result provides a graphical interpretation of a special case of the conditions \eqref{eq:iqc_decentralised_network}--\eqref{eq:iqc_decentralised_buses}.}
\begin{corollary} \label{cor:dw_shell}
    \lhf{Assume $Y_N(s)$ is invertible.\footnote{\lhf{\Cref{cor:dw_shell} can also be \lhg{extended to the case} when $Y_N(s)$ is singular~\cite{lestas_LargeScale_12}.}} Then, the conditions \eqref{eq:iqc_decentralised_network}--\eqref{eq:iqc_decentralised_buses} hold at $s$ and $\nu_i$ \lhh{for all $\tau \in [0,1]$} with}
    \begin{equation} \label{eq:iqc_dw_shell}
        \lhf{\Pi(s) = \begin{bmatrix}
            a(s) & b(s)^\ast \\ b(s) & c(s)
        \end{bmatrix} \otimes I_{2N_B},}
    \end{equation}
    \lhf{for \lhh{scalars} $a(s) \leq 0$, $b(s) \in \mathbb{C}$, and $c(s) \geq 0$ if}
    \begin{equation*} 
        \lhf{\mathcal{DW}(Z_i(s)) \ \cap \ \mathcal{DW}(-Y_N(s)^{-1}) = \varnothing}
    \end{equation*}
    \lhf{where $\mathcal{DW}(\cdot)$ is the DW shell defined in \cref{def:dw_shell}}.
\end{corollary}
\begin{proof}
    \lhf{This follows from \cite[Theorem 3.1]{lestas_LargeScale_12}} \lhh{and \cref{cor:iqc_convex}}. 
\end{proof}

\lhf{Due to the convexity of the DW shell, the scalars $a(s)$, $b(s)$, $c(s)$ parametrise a plane in three-dimensional space that strictly separates the DW shells of $Z_i(s)$ and $-Y_N(s)^{-1}$ \cite{lestas_LargeScale_12}. For fixed networks, \cref{cor:dw_shell} provides a graphical method to check condition \eqref{eq:iqc_decentralised_network}--\eqref{eq:iqc_decentralised_buses} at a particular $s \in \Gamma_N^{\, \epsilon}$, which can facilitate the identification of multipliers.
}
\begin{remark}
    \lhf{As we will see in \cref{cor:pos_real,cor:fixed_net}, many standard results in the literature, such as the positive-real, small-gain, and small-phase conditions \cite{huang_GainPhase_24}, can be written as constraints with multipliers of the form \eqref{eq:iqc_dw_shell}. The DW shell separation condition has been further exploited in \cite{huang_GeometricDecentralized_25,leng_GeometricDecentralized_26} to include graphical verification using two-dimensional projected versions of the DW shell for networks with homogeneous R/X ratios, such as scaled relative graph (SRG) separation. However, the full-block nature of \eqref{eq:multiplier_decentralised} contains additional degrees of freedom compared to \eqref{eq:iqc_dw_shell}, meaning \cref{lemma:iqc_decentralised} offers further generality to design conditions that cannot be captured by a separating plane in three dimensions, while still retaining decentralised constraints.}
\end{remark}

We now specialise \cref{lemma:iqc_decentralised} to derive a positive-real constraint on the bus dynamics. 

\subsubsection{Positive-Real Condition}
The following condition results from applying \cref{lemma:iqc_decentralised} to the grid system $[Z_B(s),Y_N(s)]$ under \lhf{a positive-real assumption. As noted in \cref{sec:model_network}, this holds if all blocks in $Y_N(s)$ are modelled as \eqref{eq:line_admittance}--\eqref{eq:shunt_admittance} and the dynamics of any non-passive loads are included in the bus impedances}. 

\begin{corollary} \label{cor:pos_real}
    \lhf{Assume $Y_N(s)$ is positive-real (see \cref{def:positive_real})}. 
    Then, for a particular $s \in \Gamma_{N}^{\, \epsilon}$, \eqref{eq:loci_condition} holds if inequality \eqref{eq:iqc_decentralised_buses} is satisfied at every $\nu_i \in \mathcal{V}_B$ with 
    \begin{align}
        \Pi_{11}^i(s) = 0_2, && \Pi_{12}^i(s) = I_2, && \Pi_{22}^i(s) = 0_2.
    \end{align} 
\end{corollary}
\begin{proof}
    By \cref{def:positive_real}, $Y_N(s) + Y_N(s)^\ast \geq 0$ for all $s \in \Gamma_N^{\, \epsilon}$, which can be written in the form \eqref{eq:iqc_decentralised_network} for $\tau = 1$. The result then follows from \cref{cor:iqc_convex,lemma:iqc_decentralised}. 
\end{proof}

The condition in \cref{cor:pos_real} can be viewed as a \lhf{frequency-wise} positive-real constraint on the bus dynamics at a particular $s \in \Gamma_{N}^{\, \epsilon}$. We also have the following graphical interpretation of the conditions in \cref{cor:pos_real}. 
\begin{corollary} \label{cor:pos_real_nr}
    \lhf{Assume $Y_N(s)$ is positive-real}. 
    Then, for a particular $s \in \Gamma_{N}^{\, \epsilon}$, \eqref{eq:loci_condition} holds if the following condition is satisfied at $\nu_i \in \mathcal{V}_B$: 
    \begin{equation*} 
            \mathcal{W}(Z_i(s)) \subset \mathbb{C}_+,
    \end{equation*}
    where $\mathcal{W}(\cdot)$ is the numerical range in \cref{def:numerical_range}. 
\end{corollary}
\begin{proof}
    This follows from \cref{lemma:numerical_range_pos_def,cor:pos_real}. 
\end{proof}

As \lhf{$Y_N(s)$ composed of \eqref{eq:line_admittance}--\eqref{eq:shunt_admittance} is positive-real} \lhg{for any topology and for all network parameters}, the constraint in \cref{cor:pos_real} (equivalently \cref{cor:pos_real_nr}) can be checked independently by each bus in the network, without knowledge of the dynamics of any other bus or the full network admittance matrix, enabling decentralised verification and plug-and-play functionality. \lhg{However, on its own, it is a conservative condition that is typically not satisfied at all frequency ranges.}

\subsubsection{Infinity-Norm Condition}
\lhb{%
As in the case of DC grids \cite{laib_DecentralizedStability_23}, a matrix infinity-norm condition can be derived for AC grids that allows for plug-and-play functionality. 
}%
\begin{corollary} \label{cor:inf_norm}
    \lhe{Let \cref{assump:connected} hold} and let each block of $Y_N(s)$ in \eqref{eq:network_admittance_blocks} be modelled as \eqref{eq:line_admittance} \lhe{or \eqref{eq:shunt_admittance}} and for each $\nu_i \in \mathcal{V}_B$, define
    \begin{equation*} 
        \gamma_i(s) = \lVert Y_i(s) + \textstyle\sum_{\nu_k \in \mathcal{N}_i} Y_{ik}(s) \rVert_\infty + \sum_{\nu_k \in \mathcal{N}_i} \lVert Y_{ik}(s)\rVert_\infty,
    \end{equation*}
    where $\lVert \cdot \rVert_\infty$ is the induced matrix $\infty$-norm (maximum absolute row-sum \cite{horn_MatrixAnalysis_85}).
    Then, for a particular $s \in \Gamma_{N}^{\, \epsilon}$, \eqref{eq:loci_condition} holds if inequality \eqref{eq:iqc_decentralised_buses} is satisfied at every $\nu_i \in \mathcal{V}_B$ with 
    \begin{align*}
        \Pi_{11}^i(s) = -\gamma_i(s)I_2, && \Pi_{12}^i(s) = 0_2, && \Pi_{22}^i(s) = \gamma_i(s)^{-1}I_2.
    \end{align*}
\end{corollary}
\begin{proof}
    \lhb{The corollary is proved in Appendix \ref{app:inf_norm_proof}.}
\end{proof}
\lhb{%
\cref{cor:inf_norm} thus provides a locally verifiable condition on the bus impedance that depends only on the parameters of the lines adjacent to bus $\nu_i$ and any local load/shunt. Therefore, each bus can independently check the condition over a specified frequency range without knowledge of the dynamics of other buses or the full network admittance matrix.
}%

\subsection{Decentralised Extended-Subsystem Conditions} \label{sec:stab_extended}
In this section, using ideas from \cite{lestas_NetworkStability_11, pates_ScalableDesign_17, laib_DecentralizedStability_23}, we demonstrate how the grid system $[Z_B(s),Y_N(s)]$ can be reformulated into an equivalent system composed of a block-diagonal set of subsystems that each contain the dynamics of a single bus and all the lines, loads, and shunts associated with that bus, connected in feedback with an interconnection matrix. We then show that quadratic constraints can be imposed on these local subsystems that collectively guarantee \eqref{eq:loci_condition} holds at a particular $s \in \Gamma_{N}^{\, \epsilon}$. 
\lhf{These conditions depend only on local dynamics and parameters, making them suitable for plug-and-play operation. }

\subsubsection{Reformulated System}
We first define some notation that will be useful in the derivation and the subsequent sections.
\begin{itemize}

    \item \lhe{Recall the edge set $\mathcal{E}_N = \mathcal{E}_P \cup \mathcal{E}_0$ from \cref{sec:model_topology}, containing the $N_P$ power lines and the loads/shunts present in the network, and let $N_E := |\mathcal{E}_N|$.} Define the \textit{extended incidence matrix} \lhe{$B_E \in \mathbb{R}^{N_B \times N_E}$ as the incidence matrix of $(\mathcal{V}_N, \mathcal{E}_N)$ with the row corresponding to $\nu_0$ removed, i.e.,}
    \begin{equation*} 
        B_E := \begin{bmatrix} B_N & S_0 \end{bmatrix},
    \end{equation*}
    where $B_N$ is defined in \eqref{eq:incidence_matrix} \lhe{and the columns of $S_0 \in \mathbb{R}^{N_B \times |\mathcal{E}_0|}$ are the standard basis vectors $e_i \in \mathbb{R}^{N_B}$ for each load/shunt $(\nu_i,\nu_0) \in \mathcal{E}_0$.} Let $\mathcal{B}_E := B_E \otimes I_2$, with $i^{\mathrm{th}}$ block-row $\mathcal{B}_E^{i\bullet}$ (of dimension $2 \times 2N_E$), \lhe{and let $Y_E(s) := \oplus_{(\nu_i,\nu_j) \in \mathcal{E}_N} Y_{ij}(s)$, with the block order following the columns of $B_E$ (recall that $Y_{ij}(s) \equiv Y_i(s)$ for loads/shunts).}

    \item \lhe{For each $\nu_i \in \mathcal{V}_B$, let $m_i$ denote the number of edges in $\mathcal{E}_N$ incident to $\nu_i$, and let $K_i \in \mathbb{R}^{N_E \times m_i}$ collect the standard basis vectors $e_j \in \mathbb{R}^{N_E}$ associated with these edges, i.e.,}
    \begin{equation*}
        K_i := \begin{bmatrix}
            e_{j_1} & e_{j_2} & \cdots  & e_{j_{m_i}}
        \end{bmatrix},
    \end{equation*}
    where $\{j_1, j_2,\ldots, j_{m_i}\}$ are the indices associated with the non-zero columns of $B_E^{i\bullet}$.
    Then, \lhg{for any $M \in \mathbb{C}^{N_E \times N_E}$,} $K_i^T M K_i$ is the sub-matrix of $M$ whose rows and columns correspond to the edges at $\nu_i$. Define $\mathcal{K}_i := K_i \otimes I_2$ and
    \begin{equation} \label{eq:K_mat_def}
        \mathcal{K} := \begin{bmatrix}
            \mathcal{K}_1 & \mathcal{K}_2 & \cdots & \mathcal{K}_{N_B}
        \end{bmatrix}.
    \end{equation}
    \lhe{Note that
    \begin{equation} \label{eq:nE_def}
        n_E := \mathcal{K}\mathcal{K}^T = \Bigl( \oplus_{(\nu_i,\nu_j) \in
        \mathcal{E}_N}\, n_{ij} \Bigr) \otimes I_2,
    \end{equation}
    where $n_{ij} = 2$ if $(\nu_i,\nu_j) \in \mathcal{E}_P$ (each line touches two buses) and $n_{ij} = 1$ if $(\nu_i,\nu_j) \in \mathcal{E}_0$.}

    \item For each $\nu_i \in \mathcal{V}_B$, define
    \begin{equation} \label{eq:E_sys_def}
        E_i(s) := Z_{E,i}(s) Y_{E,i}(s),
    \end{equation}
    where
    \begin{subequations} \label{eq:ZE_LE_def}
        \begin{align}
            Z_{E,i}(s) &:= \mathcal{K}_i^T (\mathcal{B}_E^{i\bullet})^T
            Z_i(s) \mathcal{B}_E^{i\bullet} \mathcal{K}_i, \label{eq:ZE_LE_def_Z} \\
            Y_{E,i}(s) &:= \mathcal{K}_i^T Y_E(s) \mathcal{K}_i, \label{eq:ZE_LE_def_Y}
        \end{align}
    \end{subequations}
    with $Z_i(s)$ defined in \eqref{eq:bus_iv_relationship}, and let
    $E(s) := \oplus_{\nu_i \in \mathcal{V}_B} E_i(s)$.
\end{itemize}

The structure of $E_i(s)$ --- the bus impedance $Z_i(s)$ in signed, repeated blocks, multiplied by the block-diagonal admittances of all branches at $\nu_i$ --- is illustrated in \cref{fig:tikz_extended}.
We note that each $E_i(s) \in \mathbf{RH}_{\infty}^{2m_i \times 2m_i}$ by \cref{assump:open_loop}.A.\footnote{
\lhe{Note that the blocks of $E_i(s)$ contain transfer functions that appear within blocks of $L_N(s)$ when taking the product $Z_B(s)Y_N(s)$ (up to a sign), so $E_i(s)$ is proper and stable if $L_N(s)$ is proper and stable.}
} 
With this notation, we can now prove the following result. 

\begin{figure}[t]
    \centering
    \begin{subfigure}{0.48\textwidth}
        \centering
        \input{Figures/tikz_extended_graph}
        \caption{Example graph with 4 buses and 3 power lines.}
        \label{fig:tikz_extended_graph}
    \end{subfigure}
    
    \begin{subfigure}{0.48\textwidth}
        \centering
        \input{Figures/tikz_extended_block_diag}
        \caption{The structure of $Z_{E,i}(s)$ and $Y_{E,i}(s)$ for node $\nu_2$.}
        \label{fig:tikz_extended_block_diag}
    \end{subfigure}
    \caption{Example graph to show the structure of $Z_{E,i}(s)$ and $Y_{E,i}(s)$, as defined in \eqref{eq:ZE_LE_def}, for a bus connected to two power lines and one load. 
    }
    \label{fig:tikz_extended}
\end{figure}

\begin{lemma} \label{lemma:reform_extended}
    Consider the grid system $[Z_B(s),Y_N(s)]$ \lhe{described in \cref{sec:model} and} illustrated in \cref{fig:tikz_model_iv}.
    For a particular $s \in \Gamma_{N}^{\, \epsilon}$, the condition \eqref{eq:loci_condition} holds if and only if
    \begin{equation} \label{eq:loci_condition_extended}
        -1 \notin \sigma(\tau  E(s) \mathcal{K}^T\mathcal{K}), \quad \forall \tau \in [0,1],
    \end{equation}
    where $E(s)$ and $\mathcal{K}$ \lhh{are} defined above.  
\end{lemma}
\begin{proof}
    The lemma is proved in Appendix \ref{app:extended_reform_proof}. 
\end{proof}

This result indicates that stability of the grid system $[Z_B(s),Y_N(s)]$ can equivalently be determined by analysing the feedback system $[E(s),\mathcal{K}^T \mathcal{K}]$.

\subsubsection{Stability Conditions on Reformulated System}
The system $E(s)$ is a block-diagonal system (with structure of each block similar to the example in \cref{fig:tikz_extended}), while $\mathcal{K}^T\mathcal{K}$ is a static matrix with entries of $0$ and $1$ corresponding to the sparsity of the network. This sparsity can be exploited to design decentralised conditions \lhf{with the plug-and-play property} that can be used to certify that \eqref{eq:loci_condition_extended} (and therefore \eqref{eq:loci_condition}) holds.

\begin{lemma} \label{lemma:iqc_ext}
    Consider the reformulated grid system $[E(s), \mathcal{K}^T \mathcal{K}]$ \lhe{defined above}. For a particular $s \in \Gamma_{N}^{\, \epsilon}$ and for each $(\nu_i,\nu_j) \in \mathcal{E}_N$, let $\pi_{12}^{ij}(s) \in \mathbb{C}^{2\times2}$ and $\pi_{22}^{ij}(s) \in \mathbb{H}^2$, with $\pi_{22}^{ij}(s) \geq 0$, be multiplier blocks such that the following matrix inequality holds: 
    \begin{equation} \label{eq:iqc_ext_condition}
        -\pi_{12}^{ij}(s) - \pi_{12}^{ij}(s)^\ast + n_{ij}\pi_{22}^{ij}(s) \leq 0,
    \end{equation}
    \lhe{where $n_{ij}$ is as in \eqref{eq:nE_def}.}
    Let $\pi_{12}(s) = \oplus_{(\nu_i,\nu_j) \in \mathcal{E}_N} \pi_{12}^{ij}(s)$ and $\pi_{22}(s) = \oplus_{(\nu_i,\nu_j) \in \mathcal{E}_N} \pi_{22}^{ij}(s)$. Then, for each $\nu_i \in \mathcal{V}_B$, take
    \begin{subequations} \label{eq:mult_ext_bus_level}
        \begin{align} 
            \Pi_{12}^i(s):= \mathcal{K}_i^T \pi_{12}(s) \mathcal{K}_i,  \\  
            \Pi_{22}^i(s):= \mathcal{K}_i^T \pi_{22}(s) \mathcal{K}_i,
        \end{align}
    \end{subequations}
    i.e., the blocks of $\Pi_{12}^i(s)$ and $\Pi_{22}^i(s)$ are selected to include the multiplier terms associated with the edges connected to $\nu_i \in \mathcal{V}_B$. 
    Then, condition \eqref{eq:loci_condition_extended} holds if, for every $\nu_i \in \mathcal{V}_B$,
    \begin{equation} \label{eq:iqc_ext_local}
        \begin{bmatrix}
            E_i(s) \\ I_{2m_i}
        \end{bmatrix}^\ast \begin{bmatrix}
            0 & \ \Pi_{12}^i(s)^\ast \\
            \Pi_{12}^i(s) &  \Pi_{22}^i(s)
        \end{bmatrix} \begin{bmatrix}
            E_i(s) \\ I_{2m_i}
        \end{bmatrix} > 0.
    \end{equation} 
\end{lemma}
\begin{proof}
    The lemma is proved in Appendix \ref{app:iqc_ext_proof}. 
\end{proof}
\Cref{lemma:iqc_ext} suggests the following procedure 
\lhe{to verify \eqref{eq:loci_condition_extended}:}
\begin{enumerate}
    \item At each line, load, and shunt $(\nu_i,\nu_j)\in \mathcal{E}_N$, choose local matrices $\pi_{12}^{ij}(s)$ and $\pi_{22}^{ij}(s)$ that satisfy \eqref{eq:iqc_ext_condition}. These can depend on local parameters to each line, or be set to global values as part of a network-wide grid code. 
    \item Then, for each bus $\nu_i \in \mathcal{V}_B$, check whether \eqref{eq:iqc_ext_local} is satisfied using the multiplier terms associated with the adjoining edges $\pi_{12}^{ij}(s)$ and $\pi_{22}^{ij}(s)$.
\end{enumerate}
Crucially, \lhf{as the matrices} $\pi_{12}^{ij}(s)$ and $\pi_{22}^{ij}(s)$ \lhf{are local to each line}, \cref{lemma:iqc_ext} can be used as a decentralised stability condition that enables plug-and-play functionality. 

\subsubsection{Example Conditions}
We now provide a specific example of a constraint
\lhf{for use in \cref{lemma:iqc_ext}.}

\begin{corollary} \label{cor:ext_stability_simp}
    Let $\bar{\eta}(s) \in \mathbb{R}$. Then, for a particular $s \in \Gamma_{N}^{\, \epsilon}$, \eqref{eq:loci_condition_extended} holds if inequality \eqref{eq:iqc_ext_local} is satisfied at every $\nu_i \in \mathcal{V}_B$ with \lhf{multipliers given by \eqref{eq:mult_ext_bus_level}, where, at every $(\nu_i,\nu_j) \in \mathcal{E}_N$},
    \begin{subequations} \label{eq:iqc_ext_terms_simp}
        \begin{align}  
            & \pi_{12}^{ij}(s)  = \left( 1 + j\bar{\eta}(s)  \right) \sqrt{Y_{ij}(s)^\ast} \sqrt{Y_{ij}(s)}, \\
            & \pi_{22}^{ij}(s)  = 2n_{ij}^{-1} \sqrt{Y_{ij}(s)^\ast} \sqrt{Y_{ij}(s)},
        \end{align}
    \end{subequations}
    \lhe{where $n_{ij}$ is as in \eqref{eq:nE_def}}.
\end{corollary}
\begin{proof}
    As \eqref{eq:iqc_ext_terms_simp} satisfy \eqref{eq:iqc_ext_condition}, \cref{lemma:iqc_ext} applies.
\end{proof}

\lhf{We now interpret \lhf{this constraint} graphically.}

\begin{corollary} \label{cor:ext_stability_nr}
    Let $\bar{\eta}(s) \in \mathbb{R}$ and let $\mathcal{L}_{\bar{\eta}(s)} \subset \mathbb{C}$ be the subset of the complex plane strictly to the right of a line through $-1$ defined by
    \begin{equation*} 
        \mathcal{L}_{\bar{\eta}(s)} = \{z = a+jb \ | \ a - \bar{\eta}(s)b > -1 \}.
    \end{equation*}
    \lhf{Assume $Y_{E,i}(s)$ in \eqref{eq:ZE_LE_def_Y} is invertible. Then, the condition \eqref{eq:iqc_ext_local} holds at $s$ with multiplier blocks \eqref{eq:iqc_ext_terms_simp}} if and only if 
    \begin{equation} \label{eq:ext_nr}
        \mathcal{W}\left(C_{E,i}(s)\right) \subset \mathcal{L}_{\bar{\eta}(s)},
    \end{equation}
    where 
    \begin{equation} \label{eq:ext_nr_mat}
        C_{E,i}(s) = \sqrt{n_i} \sqrt{Y_{E,i}(s)} Z_{E,i}(s) \sqrt{Y_{E,i}(s)} \sqrt{n_i},
    \end{equation}
    $\mathcal{W}(\cdot)$ is the numerical range in \cref{def:numerical_range}, $Z_{E,i}(s)$ is defined in \eqref{eq:ZE_LE_def_Z}, and $n_i := \mathcal{K}_i^T n_E \mathcal{K}_i$, with $n_E$ in \eqref{eq:nE_def}. 
\end{corollary}
\begin{proof}
    The corollary is proved in Appendix \ref{app:ext_stability_nr_proof}.
\end{proof}

The corollary demonstrates that the condition \eqref{eq:iqc_ext_local} with blocks given in \eqref{eq:iqc_ext_terms_simp} is equivalent to the condition that the numerical range of the matrix $C_{E,i}(s)$ lies strictly to the right of a line passing through the point $-1$. The angle of the line is determined by the parameter $\bar{\eta}(s)$ (setting $\bar{\eta}(s)=0$ results in a vertical line). This graphical interpretation provides a simple method for finding feasible values of $\bar{\eta}(s)$ \lhf{to facilitate the identification of multipliers for \eqref{eq:iqc_ext_local}}.

\lhf{Moreover,} the role of the $\sqrt{Y_{ij}(s)}$ factors in \eqref{eq:iqc_ext_terms_simp} also becomes clear via the graphical interpretation of \cref{cor:ext_stability_nr}: the matrix $\sqrt{Y_{E,i}(s)} \,Z_{E,i}(s) \, \sqrt{Y_{E,i}(s)}$ is typically \lhg{closer to Hermitian} than $E_i(s) = Z_{E,i}(s)Y_{E,i}(s)$,
\lhg{so its numerical range lies closer to the convex hull of its eigenvalues.}
\lhe{Condition \eqref{eq:ext_nr} is therefore less conservative than if applied to $E_i(s)$ directly.}

\lhf{We now state a less conservative condition that allows us to exploit the fact that \eqref{eq:iqc_ext_condition} is a constraint applied at each edge.}

\begin{corollary} \label{cor:ext_stability}
     For each $(\nu_i,\nu_j) \in \mathcal{E}_N$, let $\eta_{ij}(s) \in \mathbb{R}$ and $\gamma_{ij}(s) \in \mathbb{R}$. Then, for a particular $s \in \Gamma_{N}^{\, \epsilon}$, \eqref{eq:loci_condition_extended} holds if inequality \eqref{eq:iqc_ext_local} is satisfied at every $\nu_i \in \mathcal{V}_B$ with \lhf{multipliers given by \eqref{eq:mult_ext_bus_level}, where, at every $(\nu_i,\nu_j) \in \mathcal{E}_N$},
    \begin{subequations} \label{eq:iqc_ext_terms}
        \begin{align}  
            &\begin{aligned}
                \pi_{12}^{ij}(s) = \sqrt{Y_{ij}(s)^\ast} & \Bigl[ (1 + j\eta_{ij}(s) ) I_2 \Bigr. \\
                &\Bigl. + \gamma_{ij}(s) J \Bigr] \sqrt{Y_{ij}(s)},
            \end{aligned} \\
            & \pi_{22}^{ij}(s)  = \frac{2}{n_{ij}} \sqrt{Y_{ij}(s)^\ast} \sqrt{Y_{ij}(s)},
        \end{align}
    \end{subequations}
    where $J = \begin{bsmallmatrix} 0 & -1 \\ 1 & 0 \end{bsmallmatrix}$, \lhe{and $n_{ij}$ is as in \eqref{eq:nE_def}}.
\end{corollary}
\begin{proof}
    As \eqref{eq:iqc_ext_terms} satisfy \eqref{eq:iqc_ext_condition}, \cref{lemma:iqc_ext} applies.
\end{proof}

When using \eqref{eq:iqc_ext_terms}
\lhf{in \cref{lemma:iqc_ext}}, condition \eqref{eq:iqc_ext_local} imposes a constraint on the local extended subsystem $E_i(s)$ that depends only on the dynamics of local lines, loads and shunts, and on the scalars $\eta_{ij}(s)$ and $\gamma_{ij}(s)$, which \lhf{are allowed to} differ at each edge \lhf{(as \eqref{eq:iqc_ext_local} using \eqref{eq:iqc_ext_terms} are linear in the scalars $\eta_{ij}(s)$, $\gamma_{ij}(s)$, suitable values can be found at each $s \in \Gamma_N^{\,\epsilon}$ by solving a local linear matrix inequality \cite{boyd_LinearMatrix_94})}. 
\lhf{If these are infeasible, further flexibility is available by exploiting the full-block structure of each $\pi_{12}^{ij}(s)$ and $\pi_{22}^{ij}(s)$.}
The condition \eqref{eq:iqc_ext_local} therefore facilitates plug-and-play operation as the addition of a new line to the network only impacts the stability conditions imposed on any buses connected to that line. 

\lhf{Finally, we note that the sparsity of $Y_N(s)$ was the key feature that allowed us to derive \cref{lemma:reform_extended,lemma:iqc_ext}. As $\widehat{Z}_N(s) = Y_N(s)^{-1}$ is in general non-sparse, conditions on extended subsystems (and the plug-and-play benefits they offer) cannot be used in the admittance representation.}

\begin{remark} \label{remark:distributed_conds}
    The system $E(s)$ derived in this section consists of a set of extended bus subsystems that couple the dynamics of each bus to its neighbouring lines, loads and shunts. Similarly, the methods used here can be applied to create an alternative decomposition of the grid consisting of subsystems that instead couple each line to its two neighbouring buses, along with any local loads/shunts (see \cite{hallinan_PerformanceStability_25} for details).
    \lhf{Conditions imposed on such distributed subsystems} may be less conservative than the decentralised constraints derived here and therefore more suitable in some applications. 
\end{remark}

\subsection{Local Bus Conditions for Fixed Networks} \label{sec:stab_fixed}
In this section, 
\lhc{we show that many existing results in the literature are special cases of \cref{lemma:iqc_decentralised}. In each case,} 
the multiplier parameters depend on global network parameters or centralised knowledge of the properties of $Y_N(s)$, \lhe{so these conditions are suitable for fixed network topologies only}. 

\begin{corollary} \label{cor:fixed_net}
    Consider the grid system $[Z_B(s),Y_N(s)]$ with fixed network admittance $Y_N(s)$ defined in \eqref{eq:network_admittance_blocks}. Then, for a particular $s \in \Gamma_{N}^{\, \epsilon}$, condition \eqref{eq:loci_condition} holds if inequality \eqref{eq:iqc_decentralised_buses} is satisfied at every $\nu_i \in \mathcal{V}_B$ with any of the following multiplier blocks: 
    \begin{enumerate}[i)]
        \item \textbf{Frequency-wise small-gain \cite{huang_GainPhase_24}:}
        \begin{align*}
            \Pi_{11}^i(s) = -\bar{\sigma}(s)^2 I_2,  && \Pi_{12}^i(s) = 0_2, && \Pi_{22}^i(s) = I_2,
        \end{align*}
        where $\bar{\sigma}(s) := \lVert Y_N(s) \rVert_2$ (the induced matrix $2$-norm of $Y_N(s)$ at $s \in \mathbb{C}$). 

        \item \textbf{Frequency-wise small-phase \cite{huang_GainPhase_24}:}
        \lhf{\lhg{if} $Y_N(s)$ is positive-real (see \cref{def:positive_real}), then \lhg{the following multipliers can be used:}}
        \begin{align*}
            \Pi_{11}^i(s) = 0_2,  && \Pi_{12}^i(s) = e^{-j\alpha(s)}I_2, && \Pi_{22}^i(s) = 0_2,
        \end{align*}
        where $\alpha(s) \in [-\frac{\pi}{2} - \underline{\phi}(s),\frac{\pi}{2} - \bar{\phi}(s)]$, with $\bar{\phi}(s) := \sup_{z \in \mathcal{W}(Y_N(s))} \angle z$ and $\underline{\phi}(s) = \inf_{z \in \mathcal{W}(Y_N(s))} \angle z$ (here, $\angle z$ is the argument of $z \in \mathbb{C}$ and $\mathcal{W}(\cdot)$ is the numerical range in \cref{def:numerical_range}). 

        \item \textbf{Rotated positive-real \cite{vorobev_DecentralizedStability_19,chen_ExtendedFrequencyDomain_24}:}
        \lhf{\lhg{if} there exists an $\omega_r > 0$ such that for all $s \in \{j\omega: |\omega| \leq \omega_r\}$,
        \begin{equation} \label{eq:rotated_pos_real}
            (I_{N_B} \otimes J)Y_N(s) - Y_N(s)^\ast (I_{N_B} \otimes J)  \geq 0,
        \end{equation}
        then} \lhg{the following multipliers can be used:}
        \begin{align*}
            \Pi_{11}^i(s) = 0_2,  && \Pi_{12}^i(s) = -J, && \Pi_{22}^i(s) = 0_2,
        \end{align*}
        where $J = \begin{bsmallmatrix}
            0 & -1 \\ 1 & 0
        \end{bsmallmatrix}$ and $s = j\omega$ with $|\omega| \leq \omega_r$. 

        \item \textbf{Dynamic weighted passivity \cite{gorbunov_DynamicPassivity_26}:}
        \begin{align*}
            \Pi_{11}^i(s) = 0_2,  && \Pi_{12}^i(s) = m(s), && \Pi_{22}^i(s) = 0_2,
        \end{align*}
        where $m(s) \in \mathbf{RH}_{\infty}^{2\times 2}$ is a transfer function such that $(I_{N_B} \otimes m(s))^\ast Y_N(s)$ is positive-real. 
    \end{enumerate}
\end{corollary}
\begin{proof}
    The corollary is proved in Appendix \ref{app:fixed_net_proof}. 
\end{proof}

\lhf{The property \eqref{eq:rotated_pos_real} is based on the results of \cref{lemma:rotated_admittance} for networks composed of blocks modelled as \eqref{eq:line_admittance}. However, as discussed in \cref{sec:model_network}, \eqref{eq:rotated_pos_real} may not hold in the presence of shunts modelled as \eqref{eq:shunt_admittance}, and therefore this property depends on the exact network configuration.}

\lhe{Consequently, all the parameters $\bar{\sigma}(s)$, $\alpha(s)$, $\omega_r$, and $m(s)$ require}
centralised knowledge of the properties of $Y_N(s)$. As these parameters may change any time the network is modified, they are only suitable for fixed network topologies and therefore cannot be used for plug-and-play functionality.
\lhf{Nevertheless, such conditions may be useful in cases where the topology does not change often.}

\lhe{The conditions in \cref{cor:fixed_net} correspond to a specific choice of the multiplier blocks in \eqref{eq:iqc_decentralised_buses}, with i)--iii) employing blocks proportional to $I_2$ or $J$. \Cref{lemma:iqc_decentralised} instead permits arbitrary $2\times 2$ blocks. Conic combinations of the conditions above are the simplest example of this additional freedom.}

\begin{corollary} \label{cor:conic}
    \lhc{\lhe{Consider the grid system $[Z_B(s),Y_N(s)]$ with fixed network admittance $Y_N(s)$ defined in \eqref{eq:network_admittance_blocks}.}
    Assume, for a particular $s \in \Gamma_{N}^{\, \epsilon}$, there exist scalars $c_k(s) \geq 0$, $k \in \{1,\ldots,5\}$, such that \eqref{eq:iqc_decentralised_network} holds at $\tau = 1$ with 
    \begin{subequations} \label{eq:conic_condition}
        \begin{align}
            \Pi_{11}^i(s) &= [-c_2(s) \gamma_i(s) - c_3(s) \bar{\sigma}(s)^2 ] I_2, \\
            \Pi_{12}^i(s) &= [c_1(s) + c_4(s) e^{-j\alpha(s)}]I_2 - c_5(s) J, \\
            \Pi_{22}^i(s) &= [c_2(s)\gamma_i(s)^{-1} + c_3(s)] I_2,
        \end{align}
    \end{subequations}
    where $\gamma_i(s)$, $\bar{\sigma}(s)$ and $\alpha(s)$ are as defined in \cref{cor:inf_norm,cor:fixed_net}.
    Then, condition \eqref{eq:loci_condition} holds if inequality \eqref{eq:iqc_decentralised_buses} is satisfied at every $\nu_i \in \mathcal{V}_B$.}
\end{corollary}
\begin{proof}
    \lhc{As $c_k(s) \!\geq\! 0$, \cref{cor:iqc_convex,lemma:iqc_decentralised} apply.}
\end{proof}

\lhe{Both \eqref{eq:iqc_decentralised_network} (at $\tau=1$) and
\eqref{eq:iqc_decentralised_buses} \lhf{using \eqref{eq:conic_condition}} are linear in the scalars $c_k(s)$, so suitable values can be found at each $s \in \Gamma_N^{\,\epsilon}$ by solving a linear matrix inequality \cite{boyd_LinearMatrix_94}. Further flexibility can also be obtained by allowing the scalars to differ at each $\nu_i \in \mathcal{V}_B$ (provided such a multiplier satisfies \eqref{eq:iqc_decentralised_network}).}

\lhe{Using such a conic combination greatly reduces conservatism compared to using each of the multipliers on their own, as shown in \cref{sec:case_study} where \eqref{eq:conic_condition} is used to certify stability in frequency ranges where every individual condition fails.
Note that \eqref{eq:iqc_decentralised_network} may hold with the blocks \eqref{eq:conic_condition} even when the individual condition does not --- for example, it may hold with $c_5(s) > 0$
beyond $\omega_r$ as the contributions of the remaining terms
absorb the indefinite rotated term --- further reducing conservatism.}

\begin{remark}
    Similarly to \cref{cor:pos_real_nr,cor:ext_stability_nr}, the small-phase condition has the graphical interpretation that $\mathcal{W}(e^{j\alpha(s)}Y_N(s)) \subset \bar{\mathbb{C}}_+$ and $\mathcal{W}(e^{-j\alpha(s)}Z_B(s)) \subset \mathbb{C}_+$. This can be interpreted as a MIMO small-phase condition \cite{chen_PhaseTheory_24,huang_GainPhase_24}. 
\end{remark}

\subsection{Combined Stability Criteria} \label{sec:stab_comb}
We now combine the results of the previous sections to define a set of criteria that certify the small-signal stability of the closed-loop grid system $[Z_B(s),Y_N(s)]$. 

\begin{proposition} \label{prop:stability}
    Consider the grid system $[Z_B(s),Y_N(s)]$ described in \cref{sec:model} and illustrated in \cref{fig:tikz_model_iv}, where $Z_B(s)$ and $Y_N(s)$ are given in \eqref{eq:kvl_matrix_form} and \eqref{eq:network_admittance}, respectively. Let \cref{assump:connected,assump:open_loop}.A hold. Then $[Z_B(s),Y_N(s)]$ is stable under \cref{def:stability} if
    there exists an $\bar{\epsilon}>0$ such that for all $\epsilon \in (0,\bar{\epsilon}]$
    at least one of the following conditions is satisfied at each \lhf{$s = j\omega$, $\omega \geq \epsilon$}: 
    \begin{itemize}
        \item \lhf{for every $\tau \in [0,1]$,} there exists a multiplier $\Pi(s) \in \mathbb{C}^{4N_B \times 4N_B}$ of the form \eqref{eq:multiplier_decentralised} \lhe{such that} \eqref{eq:iqc_decentralised_network} holds and \eqref{eq:iqc_decentralised_buses} holds at each bus $\nu_i \in \mathcal{V}_B$.
        \item for each $(\nu_i,\nu_j) \in \mathcal{E}_N$, there exists $\pi_{12}^{ij}(s) \in \mathbb{C}^{2 \times 2}$ and $\pi_{22}^{ij}(s) \in \mathbb{H}^2$ with $\pi_{22}^{ij}(s) \geq 0$ such that \eqref{eq:iqc_ext_condition} is satisfied and, at each bus $\nu_i \in \mathcal{V}_B$, \eqref{eq:iqc_ext_local} holds.
        \item condition \eqref{eq:loci_condition} holds.
    \end{itemize}
\end{proposition}
\begin{proof}
    \lhf{Each condition guarantees \eqref{eq:loci_condition} on $\Gamma_{j\omega}^{\epsilon+}$ by \cref{cor:iqc_convex,lemma:iqc_decentralised,lemma:iqc_ext}. Since $L_N(s)$ is real-rational, $\sigma(\tau L_N(-j\omega)) = \sigma(\tau L_N(j\omega))^\ast$ and \eqref{eq:loci_condition} also holds on $\Gamma_{j\omega}^{\epsilon-}$. On the outer arc $\Gamma_R$, properness of $L_N(s)$ gives $\sigma(\tau L_N(s)) \to \sigma(\tau L_N(\infty))$ as $R \to \infty$, which coincides with the limit of the loci along $\Gamma_{j\omega}^{\epsilon+}$ as $\omega \to \infty$, so \eqref{eq:loci_condition} extends to $\Gamma_R$. Finally, \lhh{at $s=0$,} $(I+L_N(s))^{-1}$ has at most \lhg{one} pole, so for $\epsilon$ small enough the characteristic loci avoid $-1$ on $\Gamma_\epsilon$ and \eqref{eq:loci_condition} holds there. Hence stability follows from \cref{lemma:nyquist_sufficient}.}
\end{proof}

The result follows the general form of \cref{lemma:iqc_decentralised,lemma:iqc_ext} \lhf{applied over the frequency domain}. \lhe{In practice, it is}
\cref{cor:pos_real,cor:inf_norm,cor:ext_stability_simp,cor:ext_stability,cor:fixed_net,cor:conic} that provide
\lhe{implementable conditions and form the basis of grid codes.}
Where decentralised conditions are infeasible,
\lhe{\cref{prop:stability} includes the centralised condition \eqref{eq:loci_condition} to cover problematic regions of the Nyquist contour.}

\Cref{prop:stability} requires \cref{assump:open_loop}.A to hold, which \lhe{in practice requires} each $Z_i(s)$ to be stable. As we will show in Part II of this paper, this assumption may not hold in many common network configurations. As discussed in \cref{sec:model_closed_loop}, if each $\widehat{Y}_i(s) = Z_i(s)^{-1}$ is stable, we can instead apply the Nyquist criterion to the open-loop transfer function $\widehat{L}_N(s) = \widehat{Y}_B(s)\widehat{Z}_N(s)$, with $\widehat{Y}_B(s)$ and $\widehat{Z}_N(s)$ defined in \eqref{eq:inverse_systems}. We now provide an alternative stability result for the \lhe{admittance representation}. 

\begin{proposition} \label{prop:stability_inv}
    Consider the grid system $[\widehat{Y}_B(s),\widehat{Z}_N(s)]$ described in \cref{sec:model}, where $\widehat{Y}_B(s)$ and $\widehat{Z}_N(s)$ are given in \eqref{eq:inverse_systems}. Let \cref{assump:connected,assump:open_loop}.B hold. Then $[\widehat{Y}_B(s),\widehat{Z}_N(s)]$ is stable under \cref{def:stability} if
    there exists an $\bar{\epsilon}>0$ such that for all $\epsilon \in (0,\bar{\epsilon}]$
    at least one of the following conditions is satisfied at each \lhf{$s = j\omega$, $\omega \geq \epsilon$}: 
    \begin{itemize}
        \item \lhf{for every $\tau \in [0,1]$,} there exists a multiplier $\Pi(s) \in \mathbb{C}^{4N_B \times 4N_B}$ of the form \eqref{eq:multiplier_decentralised} such that
        \begin{subequations}
        \begin{equation} \label{eq:iqc_decentralised_network_inv}
            \begin{bmatrix} 
                I_{2N_B} \\ -\tau \widehat{Z}_N(s)
            \end{bmatrix}^\ast \Pi(s) \begin{bmatrix}
                I_{2N_B} \\ -\tau \widehat{Z}_N(s)
            \end{bmatrix} \leq 0,
        \end{equation}
        \begin{equation} \label{eq:iqc_decentralised_buses_inv}
            \begin{bmatrix}
                \widehat{Y}_i(s) \\ I_2
            \end{bmatrix}^\ast \begin{bmatrix}
                \Pi_{11}^i(s) & \ \Pi_{12}^i(s)^\ast \\
                \Pi_{12}^i(s) &  \Pi_{22}^i(s)
            \end{bmatrix} \begin{bmatrix}
                \widehat{Y}_i(s) \\ I_2
            \end{bmatrix} > 0,
        \end{equation}
        \end{subequations}
        \lhh{where \eqref{eq:iqc_decentralised_buses_inv} holds at every $\nu_i \in \mathcal{V}_B$},
        and $\widehat{Y}_i(s) = Z_i(s)^{-1}$.
        
        \item for $\widehat{L}_N(s) = \widehat{Y}_B(s)\widehat{Z}_N(s)$, we have
        \begin{equation} \label{eq:loci_condition_inv}
            -1 \notin \sigma(\tau \widehat{L}_N(s)), \quad \forall \tau \in [0,1].
        \end{equation}
    
    \end{itemize}
\end{proposition}
\begin{proof}
    \lhe{The proof is analogous to \cref{prop:stability}}.
\end{proof}

Notably, using the \lhe{admittance representation} of \cref{prop:stability_inv} does not include conditions on extended subsystems like those derived in \cref{sec:stab_extended}, reducing the flexibility of the stability framework. This is because the matrix $\widehat{Z}_N(s) = Y_N(s)^{-1}$ is in general non-sparse and introduces non-local coupling, meaning the network reformulation of \cref{lemma:reform_extended} cannot be performed. In Part II of this paper, we show that the bus impedances can often be stabilised through a dynamic loop transformation that translates the system to an alternative representation at low frequencies, meaning the extended stability conditions can still be applied in such cases. 

We now discuss some implications of the proposed stability framework.
\subsubsection{Reduced conservatism}
    \lhe{The conditions of \cref{prop:stability,prop:stability_inv} can be mixed along the Nyquist contour, with each condition applied over the frequency ranges where it is feasible. As demonstrated in \cref{sec:case_study}, this mixing is essential in practice: no single condition holds over the entire frequency range, but their combination certifies stability. Moreover, the full-block multipliers of \cref{lemma:iqc_decentralised} strictly generalise the conditions reported in the literature (\cref{cor:fixed_net,cor:conic}), and the extended-subsystem reformulation of \cref{lemma:iqc_ext} offers further flexibility beyond bus-level conditions.}

\subsubsection{Locality}
    \lhe{Because the conditions of \cref{lemma:iqc_decentralised,lemma:iqc_ext} act on individual buses or extended subsystems,}
    stability can be established from a bottom-up perspective, where each bus certifies its local inequalities and coordinates with the grid operator only where centralised checks are required. \lhe{The verification effort per bus is independent of}
    network size and accommodates heterogeneous device models while minimising the need for centralised computation or global model knowledge.

\subsubsection{Grid code formulation} 
    \Cref{prop:stability,prop:stability_inv} can be used to generate a ``grid code'' --- a set of design rules that must be met prior to interconnection. The grid operator specifies, \textit{a priori}, the frequency ranges over which each condition must be satisfied, and the \lhe{device} operator
    \lhe{verifies compliance using only the local dynamics, and remains}
    free to design and modify local control schemes provided
    \lhe{the stipulated conditions are not violated}.

\subsubsection{Plug-and-play operation}
    \lhe{\lhf{The conditions presented as corollaries over \crefrange{sec:stab_decentralised}{sec:stab_fixed}} differ in how network modifications affect them.}
    The decentralised conditions of \cref{sec:stab_decentralised,sec:stab_extended} enable plug-and-play operation: when a new bus is added, it must certify that it meets the grid code
    \lhe{through} local subsystem analysis. \lhe{The addition of a new line or load triggers re-verification, via \cref{cor:inf_norm,cor:ext_stability_simp,cor:ext_stability}, only at the buses it connects, and the multiplier parameters of \cref{cor:ext_stability} can be set locally per edge or globally as part of the grid code.} 
    By contrast, the conditions of \cref{cor:fixed_net,cor:conic} \lhe{depend on global network parameters,}
    \lhe{so any modification} triggers
    \lhe{re-verification across the entire} network.
    They are therefore not suited to plug-and-play use. 

\subsubsection{Practical implementation}
    \lhe{To design grid codes using}
    \cref{lemma:iqc_decentralised,lemma:iqc_ext}, the graphical interpretation provided by \lhf{\cref{cor:dw_shell,cor:pos_real_nr,cor:ext_stability_nr}} can be used to \lhe{identify feasible}
    frequency ranges and multiplier parameters.
    Where these conditions are infeasible, flexibility is
    \lhe{first available through the edge-wise} parameters $\eta_{ij}(s)$ and $\gamma_{ij}(s)$ of \cref{cor:ext_stability},
    \lhe{and then through} 
    \lhf{the scalars $c_k(s)$ of \cref{cor:conic}.} 
    \lhe{Since} 
    \lhf{the relevant constraints} are linear matrix inequalities (LMIs), 
    \lhe{suitable parameters can be computed at discrete frequency values via convex optimisation} \cite{boyd_LinearMatrix_94}. 
    Once the grid code is assigned, each bus validates 
    its conditions over the designated frequency bands, 
    \lhe{either} over a discrete frequency grid or via the generalised KYP lemma \cite{iwasaki_GeneralizedKYP_05}. 

\section{Case Study} \label{sec:case_study}
The proposed stability framework was tested on a modified 9-bus system adapted from \cite{anderson_PowerSystem_03}, shown in \cref{fig:9bus_diagram}, in which the three synchronous generators are replaced by grid-forming inverters.
The lines, loads, and shunts were modelled as admittances $Y_{ij}(s)$ using \eqref{eq:line_admittance} \lhe{and \eqref{eq:shunt_admittance}}, with \lhf{non-homogeneous} parameters from \cite{anderson_PowerSystem_03}. An additional line ($R = \SI{0.02}{\pu}$, $X = \SI{0.1}{\pu}$) connects the network to an infinite bus through a switch. \lhe{Kron reduction was then used to form $Y_N(s)$}. 

\begin{figure}[t]
\centering
\input{Figures/9bus_diagram}
\caption{Modified 9-bus system adapted from \cite{anderson_PowerSystem_03}. The three synchronous generators have been replaced by grid-forming inverters and a switchable connection to a stiff infinite bus has been added.}
\label{fig:9bus_diagram}
\end{figure}
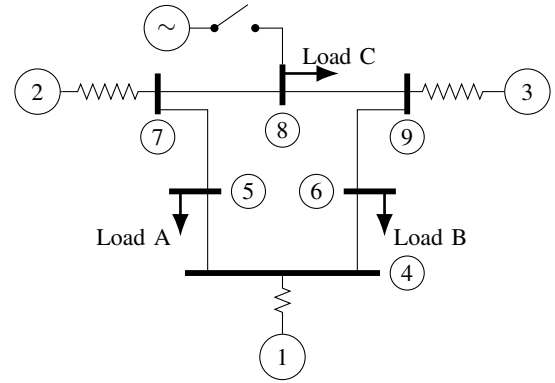

The inverters were modelled as shown in \cref{fig:inv_block_diagram} \lhf{(including dynamic, $\delta_i$-dependent reference-frame transformation \eqref{eq:dq_transform})}, with parameters given in Appendix~\ref{app:case_study_parameters}, and the system was linearised about its equilibrium point in the common $DQ$ reference frame. Although the closed-loop system was stable \lhe{(verified by centralised eigenvalue analysis)}, the inverter impedances $Z_i(s)$ at $\nu_1$ and $\nu_2$ each contained one unstable pole (as will be shown in Part II, this is because these buses have a positive reactive power output), while all the bus admittances $\widehat{Y}_i(s) = Z_i(s)^{-1}$ were stable. 

\lhf{Therefore, using the results of Part I, only \cref{prop:stability_inv} can be used to conclude stability via the admittance representation with non-sparse $\widehat{Z}_N(s) = Y_N(s)^{-1}$. This limits our ability to test conditions from \cref{sec:stab_extended} that can be used for plug-and-play operation, and we must rely on fixed-network conditions from \cref{sec:stab_fixed}. This limitation will be addressed in Part II of this paper, which removes the requirement for each $Z_i(s)$ to be stable so that conditions on extended subsystems can be used. Nevertheless, this case study allows us to test when fixed-network properties apply, as such conditions may be useful in scenarios where the grid topology does not change often, such as in microgrids.}

\subsubsection{Infinite Bus Connection}
The switch in \cref{fig:9bus_diagram} \lhf{was first closed}, \lhe{connecting the network to} an infinite bus. \lhe{\lhf{This formed a} stiff voltage source whose voltage defines the angle of the common reference frame (its $Q$-component is fixed to $0$), \lhf{so the closed-loop network did not exhibit a pole at $s=0$}}. 
In addition, the \lhe{active} power of $\nu_2$ was decreased to \SI{1.5}{\pu} \lhf{(compared to the setpoint in \cite{anderson_PowerSystem_03} of \SI{1.63}{\pu})} \lhe{so that}
power could be exchanged with the infinite bus. 

\begin{figure}[t]
    \centering
    \includegraphics[width=0.48\textwidth]{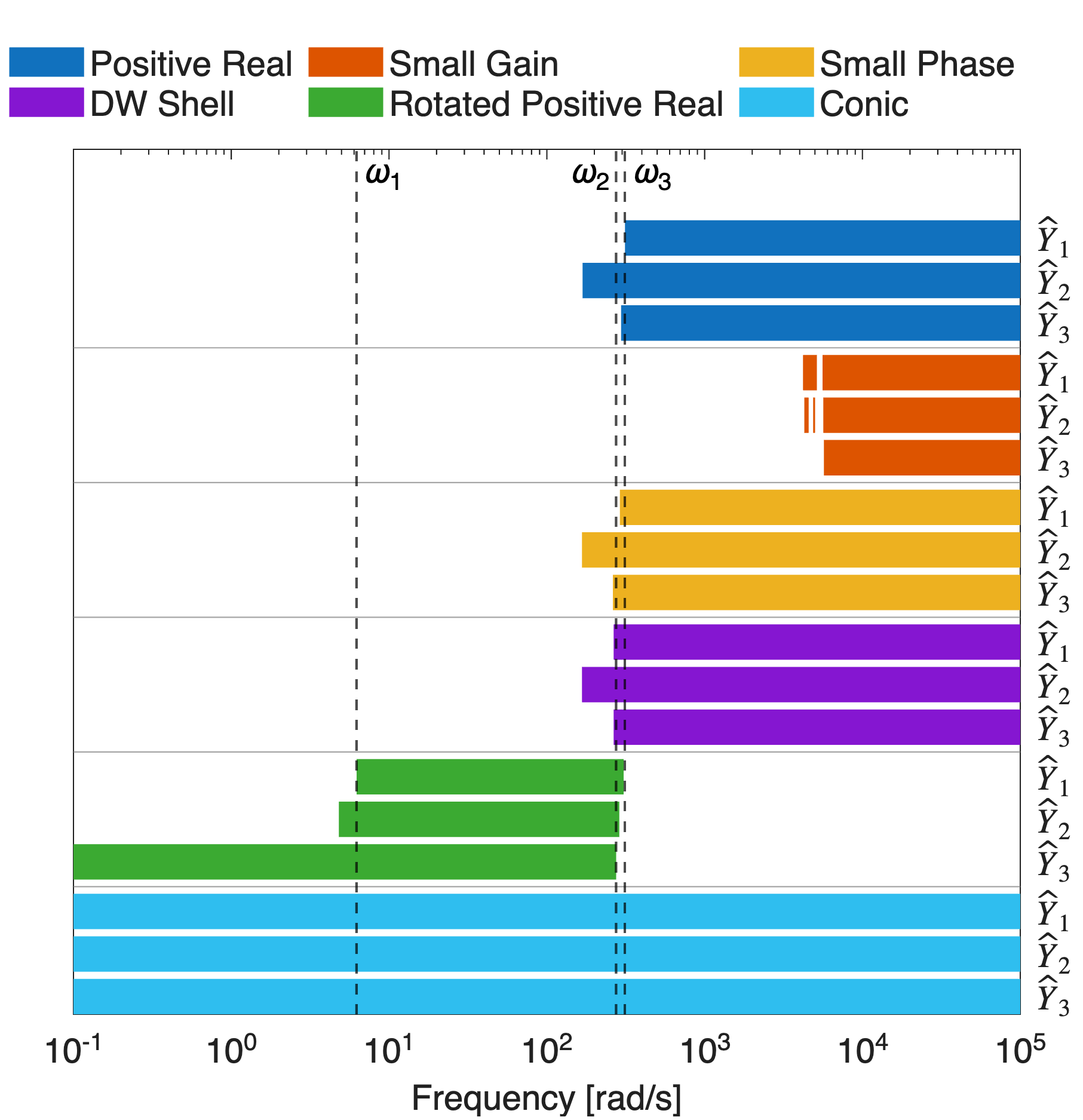}
    \caption{Frequency ranges over which the conditions of \cref{cor:pos_real,cor:dw_shell,cor:fixed_net,cor:conic} were satisfied in the case where the switch in \cref{fig:9bus_diagram} was closed.
    Here, $\omega_1 = \SI{6.2}{\radian\per\second}$, $\omega_2 = \SI{275}{\radian\per\second}$, and $\omega_3 = \SI{312}{\radian\per\second}$ represent boundaries between which conditions were satisfied by all buses.}
    \label{fig:case_study_inf}
\end{figure}

We then determined the frequency bands over which decentralised conditions of the form \eqref{eq:iqc_decentralised_buses_inv} held for each of the inverters: 
the positive-real condition of \cref{cor:pos_real}; the small-gain, small-phase, and rotated positive-real conditions of \cref{cor:fixed_net}; the DW shell separation condition of \cref{cor:dw_shell}; and the conic combination condition of \cref{cor:conic}, \lhe{each
applied to the admittance representation.\footnote{
\lhe{For the admittance representation, the conditions of \cref{cor:pos_real,cor:dw_shell,cor:fixed_net,cor:conic} apply with $Y_N(s)$ replaced by $\widehat{Z}_N(s)$ and $Z_i(s)$ by $\widehat{Y}_i(s)$: $\widehat{Z}_N(s)$ is positive-real whenever $Y_N(s)$ is strictly positive-real, and
\eqref{eq:rotated_pos_real} implies $\widehat{Z}_N(j\omega)^\ast (I_{N_B}\otimes J) - (I_{N_B}\otimes J)\widehat{Z}_N(j\omega)  \geq 0$ for $|\omega|\leq\omega_r$, so the rotated positive-real multiplier becomes $\Pi_{12}^i(s) = J$ \lhg{and the term involving $J$ in \eqref{eq:conic_condition} changes sign}.}
}} 

For this particular network \lhe{(with closed switch)}, \eqref{eq:rotated_pos_real} held below $\omega_r = \omega_0= \SI{377}{\radian\per\second}$. As shown in \cref{fig:case_study_inf}, a condition of the form \eqref{eq:iqc_decentralised_buses_inv} held at all buses across the entire frequency spectrum. 
The positive-real condition was satisfied above $\omega_3 = \SI{312}{\radian\per\second}$, while the rotated positive-real condition was satisfied between $\omega_1 = \SI{6.2}{\radian\per\second}$ and $\omega_2 = \SI{275}{\radian\per\second}$. Between $\omega_2$ and $\omega_3$, the DW shell separation condition was satisfied.
\lhf{Across the entire frequency spectrum, including below $\omega_1$ where no other single condition held}, the conic combination of \cref{cor:conic} was satisfied, demonstrating \lhe{the reduced conservatism offered by} the full-block \lhe{multiplier} structure. \lhf{This is because conic combinations allow us to use a weighted combination of constraints inspired by graphical properties that work well at high frequencies (such as the positive-real or small-gain conditions) and other conditions that work well at low frequencies (like the rotated positive-real property).} 
\lhg{In particular, the condition of \cref{cor:conic} held with $c_3(j\omega)=0.8$, $c_5(j\omega) = 0.2$, and all other $c_k(j\omega) = 0$ in \eqref{eq:conic_condition} for $\omega < \omega_1$}.

Therefore, conditions of the form \eqref{eq:iqc_decentralised_buses_inv} held at all buses over the entire frequency spectrum, and stability was concluded from \cref{prop:stability_inv} \lhe{using decentralised checks alone}. 

\subsubsection{No Infinite Bus Connection}
We then \lhe{repeated the tests with the switch open} \lhf{and with the active power of $\nu_2$ restored to \SI{1.63}{\pu}}
\lhf{Without the infinite bus, there was no absolute angle reference when defining $\omega_0$ in \eqref{eq:delta_dot}, so the closed-loop system had a pole at $s=0$ (see \cref{sec:model_closed_loop}).}

The results are shown in \cref{fig:case_study_no_inf}. \lhf{In this case,} a condition of the form \eqref{eq:iqc_decentralised_buses_inv} \lhf{only} held at all buses above $\omega_1 = \SI{8.4}{\radian\per\second}$. 
In particular, the rotated positive-real condition held between $\omega_1$ and $\omega_2 =\SI{92}{\radian\per\second}$ (with the switch open, \eqref{eq:rotated_pos_real} only held below $\omega_r = \omega_2$), while the positive-real condition held above $\omega_3 = \SI{315}{\radian\per\second}$. Only the conic combination \eqref{eq:conic_condition} held at all frequencies between $\omega_2$ and $\omega_3$, \lhg{where a combination of the positive-real and rotated positive-real conditions was used (recall that $c_5(s) > 0$ can remain feasible beyond $\omega_r$)}.

\begin{figure}[t]
    \centering
    \includegraphics[width=0.48\textwidth]{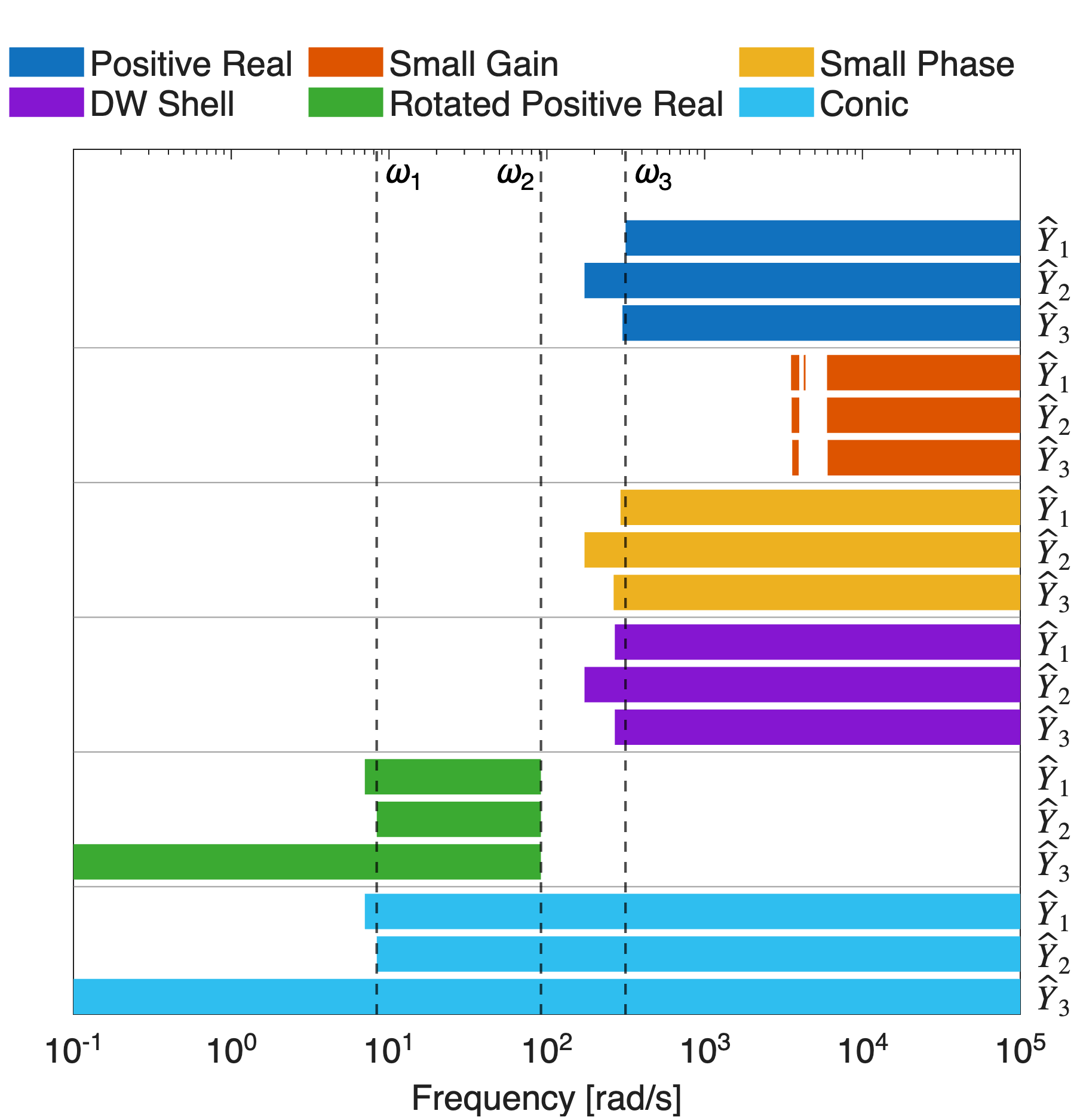}
    \caption{\lhe{As \cref{fig:case_study_inf}, with the switch in \cref{fig:9bus_diagram} open}. Here, $\omega_1 = \SI{8.4}{\radian\per\second}$, $\omega_2 = \SI{92}{\radian\per\second}$, and $\omega_3 = \SI{315}{\radian\per\second}$.}
    \label{fig:case_study_no_inf}
\end{figure}

Below $\omega_1$, centralised analysis was required. It was verified directly that \eqref{eq:loci_condition_inv} held in this range, so stability \lhe{follows from} \cref{prop:stability_inv}. 
\lhf{We believe the infeasibility of decentralised conditions at low frequencies is a fundamental feature of models lacking a stiff voltage source.}
\lhe{\lhf{This is because the} characteristic loci of $\widehat{L}_N(s)$ approach $-1$ at low frequencies due to the closed-loop pole at the origin,}
\lhf{making decentralised constraints based on strict quadratic separation of $\widehat{Y}_B(s)$ and $\widehat{Z}_N(s)$ difficult to satisfy.} 

\lhf{To the best of our knowledge, this remains an issue that has not been resolved in the literature, as other validation studies include an infinite bus or simplify the bus dynamics at low frequencies using reference-frame transformations \eqref{eq:dq_transform} with steady-state (non-dynamic) $\delta_i$ (see for example \cite{huang_GainPhase_24,huang_GeometricDecentralized_25,leng_GeometricDecentralized_26}). In Part II, we directly compute the centralised low-frequency Nyquist plot and show that, under some standard assumptions, encirclements typically do not occur in this frequency range, making the infeasibility of low-frequency decentralised stability checks non-problematic.}

\section{Conclusion} \label{sec:conclusion}
Part I of this paper presented a quadratic-constraint framework for decentralised small-signal stability studies in heterogeneous AC grids that generalises many previously reported conditions and offers increased flexibility. The conditions apply to either an impedance or admittance representation of the buses in the frequency domain and, in the impedance representation, extend to subsystems comprising a bus and all its connecting lines and loads. A modified 9-bus case study demonstrated the reduced conservatism offered by the full-block multipliers and the framework's suitability as a foundation for grid codes \lhf{in fixed networks}.

The case study also exposed two limitations motivating Part II. First, two of the bus impedances were unstable, so the impedance representation and its extended-subsystem formulation could not be used, \lhf{limiting the ability to apply conditions that allow for plug-and-play functionality}. Part II demonstrates analytically how such unstable poles arise and shows that they can be stabilised by a dynamic loop transformation to a \lhh{hybrid impedance/}power representation that better captures power-flow behaviour at low frequencies, restoring the ability to apply extended-subsystem criteria. Second, without a stiff voltage source, no decentralised condition was feasible at low frequencies. Part II shows how the Nyquist analysis in this range simplifies under the power representation, and argues that this limitation is typically not restrictive in practice.

\appendix
\subsection{Proof of \cref{lemma:nyquist_sufficient}} \label{app:nyquist_sufficient_proof}
\begin{proof}
    Since $L(s) \in \mathbf{RH}_{\infty,0}^{n\times n}$, we have $P_{ol} = 0$ in \eqref{eq:nyquist_pole_count}. Fix $\epsilon \in (0,\bar{\epsilon}]$. On $\Gamma_{N}^{\,\epsilon}$, \lhf{$\sigma(L_\tau(s))$ for} $\tau \in [0,1]$ defines a homotopy of the loci from $\{0\}$ to $\sigma(L(s))$. By \eqref{eq:nyquist_condition} the point $-1$ is avoided throughout, so the number of encirclements is invariant under the homotopy and equals $0$. \lhe{As this holds for every $\epsilon \in (0,\bar{\epsilon}]$, we have $N_\Gamma = 0$, and \cref{thm:nyquist} gives $P_{cl} = 0$, i.e., $(I+L(s))^{-1}$ has no poles in $\mathbb{C}_+$. Furthermore, evaluating \eqref{eq:nyquist_condition} for $\tau = 1$ on $\Gamma_{j\omega}^{\epsilon\pm}$ for every $\epsilon \in (0,\bar{\epsilon}]$ gives $\det(I+L(j\omega)) \neq 0$ for all $\omega \neq 0$, so there are no closed-loop poles on $j\mathbb{R}\setminus\{0\}$. Together with the assumption that, \lhh{at $s=0$,} $(I+L(s))^{-1}$ has at most \lhg{one} pole, \lhg{this gives stability under \cref{def:stability}}}.
\end{proof}

\subsection{Proof of \cref{lemma:iqc}} \label{app:iqc_proof}
\begin{proof}
    Suppose $-1 \in \sigma(\tau L(s))$ \lhh{at a given} $s \in \Gamma_N^{\, \epsilon}$ and $\tau \in [0,1]$, with $L(s) = P(s)K(s)$. Therefore, the matrix $I_n + \tau P(s)K(s)$ is singular and there exists a $w \in \mathbb{C}^n \setminus \{0\}$ such that $[I_n + \tau P(s)K(s)]w = 0$. Let $z = -\tau K(s)w$, with $z\neq 0$, so that $w - P(s) z = 0$, which gives $w = P(s) z$. 
    Now, pre- and post-multiplying \eqref{eq:iqc_conditions_top} by $z^\ast$ and $z$ and \eqref{eq:iqc_conditions_btm} by $w^\ast$ and $w$ gives 
    \begin{equation*}
        \begin{aligned}
            \begin{bmatrix}
                w \\ z
            \end{bmatrix}^\ast \Pi(s)
            \begin{bmatrix}
                w \\ z
            \end{bmatrix} > 0,   && 
            \begin{bmatrix}
                w \\ z
            \end{bmatrix}^\ast \Pi(s)
            \begin{bmatrix}
                w \\ z
            \end{bmatrix} \leq 0,
        \end{aligned}
    \end{equation*}
    which is a contradiction. Therefore, satisfaction of \eqref{eq:iqc_conditions} at a particular $s \in \Gamma_N^{\, \epsilon}$ and $\tau$ \lhh{gives $-1 \notin \sigma(\tau L(s))$}.  
\end{proof}

\subsection{Proof of \cref{cor:inf_norm}} \label{app:inf_norm_proof}
\begin{proof}
    Let $Y_N^{(i,\bullet)}(s)$ denote the $2\times 2N_B$ block-row of $Y_N(s)$ formed by rows $2i-1$ and $2i$ of $Y_N(s)$ given in \eqref{eq:network_admittance_blocks}. 
    When each $Y_{ij}(s)$ is modelled as in \eqref{eq:line_admittance} \lhe{or \eqref{eq:shunt_admittance}}, both rows of $Y_N^{(i,\bullet)}(s)$ have the same \lhe{(non-zero under \cref{assump:connected})} absolute row sum, hence $\lVert Y_N^{(i,\bullet)}(s) \rVert_\infty = \gamma_i(s)$.
    Define the block-diagonal scaling $\gamma(s) = \oplus_{i \in \mathcal{V}_B} \left( \gamma_i(s) \right) \otimes I_2$. Then each scaled block-row has unit absolute row sum, so $\lVert \gamma(s)^{-1}Y_N(s) \rVert_\infty = 1$. 
    Similarly, $\lVert \gamma(s)^{-1}Y_N(s)^\ast \rVert_\infty = 1$. Now, the induced infinity norm of a matrix gives an upper bound on its spectral radius  $\rho(\cdot)$ \cite{horn_MatrixAnalysis_85}. Consequently, 
    \begin{equation*}
        \begin{split}
            &\rho \left( \gamma(s)^{-1}Y_N(s)^\ast \gamma(s)^{-1}Y_N(s) \right) \\
            &\qquad \leq \lVert\gamma(s)^{-1}Y_N(s)^\ast\gamma(s)^{-1}Y_N(s)\rVert_\infty \\
            &\qquad \leq \lVert\gamma(s)^{-1}Y_N(s)^\ast\rVert_\infty \lVert \gamma(s)^{-1}Y_N(s)\rVert_\infty \\
            &\qquad \leq 1,
        \end{split}
    \end{equation*}
    which in turn implies $Y_N(s)^\ast\gamma(s)^{-1}Y_N(s) \leq \gamma(s)$.
    This can be written in the form \eqref{eq:iqc_decentralised_network} for $\tau = 1$, so the result follows from \cref{cor:iqc_convex,lemma:iqc_decentralised}. 
\end{proof}

\subsection{Proof of \cref{lemma:reform_extended}} \label{app:extended_reform_proof}
\begin{proof}
Consider the open-loop transfer function $L_N(s)$ at $s \in \Gamma_N^{\, \epsilon}$. In what follows, we omit the argument $s$ to simplify the notation. Using \eqref{eq:network_admittance}, we have
\begin{equation*}
    \begin{aligned}
        L_N 
        & = Z_BY_N 
        = Z_B (\mathcal{B}_N Y_P \mathcal{B}_N^T +  Y_0) \\
        & = Z_B 
        \setlength{\arraycolsep}{2pt}
        \begin{bmatrix}
            \mathcal{B}_N & \mathcal{S}_0  
        \end{bmatrix}
        \begin{bmatrix}
            Y_P & \\ & Y_{E0} 
        \end{bmatrix}
        \begin{bmatrix}
            \mathcal{B}_N^T \\ \mathcal{S}_0^T 
        \end{bmatrix} \\ 
        & = Z_B \mathcal{B}_E Y_E \mathcal{B}_E^T,
    \end{aligned}
\end{equation*}
\lhe{where $\mathcal{S}_0 = S_0 \otimes I_2$ and $Y_{E0}=\oplus_{(\nu_i,\nu_0) \in \mathcal{E}_0} Y_i$.}
Therefore, the matrices $L_N$ and $\mathcal{B}_E^T Z_B \mathcal{B}_E Y_E$ share the same non-zero eigenvalues. We now state the following properties of the selection matrix $\mathcal{K}_i$ that can be easily verified: 
\begin{equation*}
    \mathcal{B}_E^{i\bullet} \mathcal{K}_i \mathcal{K}_i^T = \mathcal{B}_E^{i\bullet}, \;
    (\mathcal{K}_i \mathcal{K}_i^T)^2 = \mathcal{K}_i \mathcal{K}_i^T, \;
    \mathcal{K}_i \mathcal{K}_i^T Y_E = Y_E \mathcal{K}_i \mathcal{K}_i^T.
\end{equation*}
This means
\begin{equation*}
    \begin{aligned}
        \mathcal{B}_E^T Z_B \mathcal{B}_E Y_E 
        & = \left[ \textstyle\sum_{\nu_i \in \mathcal{V}_B} (\mathcal{B}_{E}^{i\bullet})^T Z_i \mathcal{B}_E^{i\bullet} \right] Y_E \\
        & = \textstyle\sum_{\nu_i \in \mathcal{V}_B} \mathcal{K}_i \mathcal{K}_i^T (\mathcal{B}_{E}^{i\bullet})^T Z_i \mathcal{B}_E^{i\bullet} \mathcal{K}_i \mathcal{K}_i^T Y_E \\
        & = \textstyle\sum_{\nu_i \in \mathcal{V}_B} \mathcal{K}_i \mathcal{K}_i^T(\mathcal{B}_{E}^{i\bullet})^T Z_i \mathcal{B}_E^{i\bullet} (\mathcal{K}_i \mathcal{K}_i^T)^2 Y_E \\
        & = \textstyle\sum_{\nu_i \in \mathcal{V}_B} \mathcal{K}_i \mathcal{K}_i^T(\mathcal{B}_{E}^{i\bullet})^T Z_i \mathcal{B}_E^{i\bullet} \mathcal{K}_i \mathcal{K}_i^T Y_E \mathcal{K}_i \mathcal{K}_i^T \\
        & = \textstyle\sum_{\nu_i \in \mathcal{V}_B} \mathcal{K}_i E_i \mathcal{K}_i^T \\
        & = \mathcal{K}E \mathcal{K}^T.
    \end{aligned}
\end{equation*}
Finally, this matrix has the same non-zero eigenvalues as $E \mathcal{K}^T\mathcal{K}$. Therefore, \eqref{eq:loci_condition} holds if and only if \eqref{eq:loci_condition_extended} holds at $s \in \Gamma_N^{\, \epsilon}$ for all $\tau \in [0,1]$.
\end{proof}

\subsection{Proof of \cref{lemma:iqc_ext}} \label{app:iqc_ext_proof}
\begin{proof}
Let 
\begin{equation} \label{eq:iqc_ext_deriv0}
    \Pi(s) = \begin{bmatrix}
        0 & \Pi_{12}(s)^\ast \\ \Pi_{12}(s) & \Pi_{22}(s)
    \end{bmatrix},
\end{equation}
where $\Pi_{12}(s) = \oplus_{\nu_i \in \mathcal{V}_B} \Pi_{12}^i(s)$ and $\Pi_{22}(s) = \oplus_{\nu_i \in \mathcal{V}_B} \Pi_{22}^i(s)$. By the same argument as in \cref{lemma:iqc} and \cref{cor:iqc_convex}, \eqref{eq:loci_condition_extended} holds if \eqref{eq:iqc_conditions} is satisfied for $P(s) = E(s)$, $K(s) = \mathcal{K}^T\mathcal{K}$, and $\tau = 1$ with the multiplier $\Pi(s)$ given by \eqref{eq:iqc_ext_deriv0}. Because $E(s)$ is block diagonal, \eqref{eq:iqc_conditions_top} is ensured by \eqref{eq:iqc_ext_local}. Thus it remains to ensure
\begin{equation} \label{eq:iqc_ext_deriv01}
    \begin{bmatrix}
        I \\ - \mathcal{K}^T \mathcal{K}
    \end{bmatrix}^\ast \Pi(s) \begin{bmatrix}
        I \\ - \mathcal{K}^T \mathcal{K}
    \end{bmatrix} \leq 0.
\end{equation}

In what follows, we omit the argument $s$ to simplify the notation. Equation \eqref{eq:iqc_ext_deriv01} simplifies to
\begin{equation} \label{eq:iqc_ext_deriv1}
    -\Pi_{12}^{\ast} \mathcal{K}^T\mathcal{K} - \mathcal{K}^T\mathcal{K} \Pi_{12} +\mathcal{K}^T\mathcal{K} \Pi_{22} \mathcal{K}^T\mathcal{K} \leq 0. 
\end{equation}
Now, for each $\nu_i \in \mathcal{V}_B$, we have $\mathcal{K}_i^T \mathcal{K}_i = I_{2m_i}$ and $\mathcal{K}_i \mathcal{K}_i^T \pi_{12} = \pi_{12} \mathcal{K}_i \mathcal{K}_i^T$, where $\pi_{12}$ is the block-diagonal matrix defined in \cref{lemma:iqc_ext}. Therefore,
\begin{equation*} 
    \begin{split}
            \mathcal{K} \Pi_{12} 
            & = \mathcal{K} \left( \oplus_{\nu_i \in \mathcal{V_B}} \Pi_{12}^i \right)\\
            & = \begin{bmatrix}
                \mathcal{K}_1^T \\ \vdots \\ \mathcal{K}_{N_B}^T 
            \end{bmatrix}^T \begin{bmatrix}
                \mathcal{K}_1^T \pi_{12} \mathcal{K}_1 & & \\ & \ddots & \\ &&\mathcal{K}_{N_B}^T \pi_{12} \mathcal{K}_{N_B}
            \end{bmatrix} \\
            & = \begin{bmatrix}
                \mathcal{K}_1\mathcal{K}_1^T \pi_{12} \mathcal{K}_1 & \cdots & \mathcal{K}_{N_B}\mathcal{K}_{N_B}^T \pi_{12} \mathcal{K}_{N_B}
            \end{bmatrix} \\
            & =\pi_{12} \begin{bmatrix}
                \mathcal{K}_1 \mathcal{K}_1^T \mathcal{K}_1 & \cdots & \mathcal{K}_{N_B}\mathcal{K}_{N_B}^T\mathcal{K}_{N_B}
            \end{bmatrix} \\ 
            & = \pi_{12} \begin{bmatrix}
                \mathcal{K}_1 & \cdots & \mathcal{K}_{N_B}
            \end{bmatrix} \\ 
            & = \oplus_{(\nu_i,\nu_j) \in \mathcal{E}_N}\pi_{12}^{ij} \mathcal{K}.
    \end{split}
\end{equation*}
Similarly, $\Pi_{12}^\ast \mathcal{K}^T = \mathcal{K}^T \pi_{12}^\ast$ and $\Pi_{22} \mathcal{K}^T = \mathcal{K}^T \pi_{22}$. Using~\eqref{eq:nE_def},
\begin{equation*} 
    \begin{split}
        \mathcal{K}^T\mathcal{K} \Pi_{22} \mathcal{K}^T\mathcal{K} 
        & = \mathcal{K}^T \mathcal{K} \mathcal{K}^T \left( \oplus_{(\nu_i,\nu_j)\in\mathcal{E}_N}  \pi_{22}^{ij} \right) \mathcal{K} \\
        & = \mathcal{K}^T \left( \oplus_{(\nu_i,\nu_j)\in\mathcal{E}_N} \ n_{ij}\pi_{22}^{ij} \right) \mathcal{K}.
    \end{split}
\end{equation*}
Substitution into \eqref{eq:iqc_ext_deriv1} yields
 \begin{equation*} 
    \begin{split}
        \mathcal{K}^T \left[  \oplus_{(\nu_i,\nu_j)\in\mathcal{E}_N} \left( -{\pi_{12}^{ij}}^\ast -  \pi_{12}^{ij} + n_{ij} \pi_{22}^{ij} \right) \right] \mathcal{K} \leq 0 .
   \end{split}
\end{equation*}
This holds whenever \eqref{eq:iqc_ext_condition} holds for each $(\nu_i,\nu_j) \in \mathcal{E}_N$, since pre- and post-multiplication by $\mathcal{K}^T$ and $\mathcal{K}$ preserves negative semidefiniteness. Therefore, \eqref{eq:iqc_ext_deriv01} is satisfied, which ensures \eqref{eq:loci_condition_extended} via \cref{cor:iqc_convex}. 
\end{proof}

\subsection{Proof of \cref{cor:ext_stability_nr}} \label{app:ext_stability_nr_proof}
\begin{proof}
In what follows, we omit the argument $s$ to simplify the notation. Assume \eqref{eq:iqc_ext_local} is satisfied with 
\lhf{multiplier blocks \eqref{eq:iqc_ext_terms_simp}. Using \eqref{eq:mult_ext_bus_level} gives}
\begin{align*}  
    & \Pi_{12}^i = \left( 1 + j\bar{\eta}  \right) \sqrt{Y_{E,i}^\ast} \sqrt{Y_{E,i}}, 
    & \Pi_{22}^i  = 2n_i^{-1} \sqrt{Y_{E,i}^\ast} \sqrt{Y_{E,i}}.
\end{align*}
Therefore, \eqref{eq:iqc_ext_local} simplifies to 
\begin{equation*}
    \begin{aligned}
        (1+j\bar{\eta}) \sqrt{Y_{E,i}^\ast} \sqrt{Y_{E,i}} E_i + E_i^\ast & \sqrt{Y_{E,i}^\ast} \sqrt{Y_{E,i}} (1-j \bar{\eta}) \\
        & + 2n_i^{-1} \sqrt{Y_{E,i}^\ast} \sqrt{Y_{E,i}} > 0. 
    \end{aligned}
\end{equation*}
Using \eqref{eq:E_sys_def} and factoring out $\sqrt{n_i}^{-1}\sqrt{Y_{E,i}^\ast}$ on the left and $\sqrt{Y_{E,i}}\sqrt{n_i}^{-1}$ on the right yields
\begin{equation*}
    \begin{aligned}
        (1+j\bar{\eta}) &\sqrt{n_i}\sqrt{Y_{E,i}} Z_{E,i} \sqrt{Y_{E,i}} \sqrt{n_i}  \\
        &+ \sqrt{n_i}\sqrt{Y_{E,i}^\ast}Z_{E,i}^\ast  \sqrt{Y_{E,i}^\ast}\sqrt{n_i} (1-j \bar{\eta}) 
         + 2I_{2m_i} > 0, 
    \end{aligned}
\end{equation*}
or $(1+j\bar{\eta}) C_{E,i} + C_{E,i}^\ast (1-j \bar{\eta}) + 2I > 0$ where $C_{E,i}$ is defined in \eqref{eq:ext_nr_mat}. Take $z\in \mathbb{C}^{2m_i}$ with $\lVert z \rVert = 1$. Let $z^\ast C_{E,i}z = w = a+jb$ so we have $w \in \mathcal{W}(C_{E,i})$. Now, left- and right-multiplying our matrix inequality by $z^\ast$ and $z$ gives
\begin{equation*}
    \begin{split}
       z^\ast(C_{E,i} + C_{E,i}^\ast)z + j \bar{\eta} z^\ast(C_{E,i} - C_{E,i}^\ast)z + 2  > 0  \\
       w + w^\ast + j\bar{\eta}(w - w^\ast) + 2 > 0 \\
       a - \bar{\eta}b + 1 > 0, \\
    \end{split}
\end{equation*}
which implies $w \in \mathcal{L}_{\bar{\eta}}$. This holds for all $\lVert z \rVert = 1$, so $\mathcal{W}(C_{E,i}) \subset \mathcal{L}_{\bar{\eta}}$. For the converse, reverse the argument above, and set $z = \frac{x}{\lVert x \rVert}$ for arbitrary $x \in \mathbb{C}^{2m_i}$, which yields \eqref{eq:iqc_ext_local} with $\Pi_{12}^i$ and $\Pi_{22}^i$ containing blocks \eqref{eq:iqc_ext_terms_simp}.  
\end{proof}

\subsection{Proof of \cref{cor:fixed_net}} \label{app:fixed_net_proof}
\begin{proof}
i) By definition $Y_N(s)^\ast Y_N(s) \leq \bar{\sigma}(s)^2 I$.
ii) Under the positive-real hypothesis, for each $\alpha(s) \in [-\frac{\pi}{2} - \underline{\phi}(s),\frac{\pi}{2} - \bar{\phi}(s)]$, we have $\mathcal{W}(e^{j\alpha(s)}Y_N(s)) \subset \bar{\mathbb{C}}_+$, which, by \cref{lemma:numerical_range_pos_def}, implies $e^{j\alpha(s)}Y_N(s) + Y_N(s)^\ast e^{-j\alpha(s)} \geq 0$. 
iii) For $s \in \{j\omega: |\omega| \leq \omega_r\}$, we have \eqref{eq:rotated_pos_real} by hypothesis. 
iv) Since $(I_{N_B} \otimes m(s))^\ast Y_N(s)$ is positive-real by hypothesis, $(I_{N_B} \otimes m(s))^\ast Y_N(s) + Y_N(s)^\ast(I_{N_B} \otimes m(s)) \geq 0$. 
All the above can be written in the form \eqref{eq:iqc_decentralised_network} at $\tau = 1$ and the result follows from \cref{cor:iqc_convex,lemma:iqc_decentralised}. 
\end{proof}

\subsection{Case Study Parameters} \label{app:case_study_parameters}

\begin{figure}[t]
\centering
\input{Figures/inv_block_diagram}
\caption{Block diagram representing the grid forming inverter dynamics. Parameters are listed in Appendix \ref{app:case_study_parameters}.}
\label{fig:inv_block_diagram}
\end{figure}
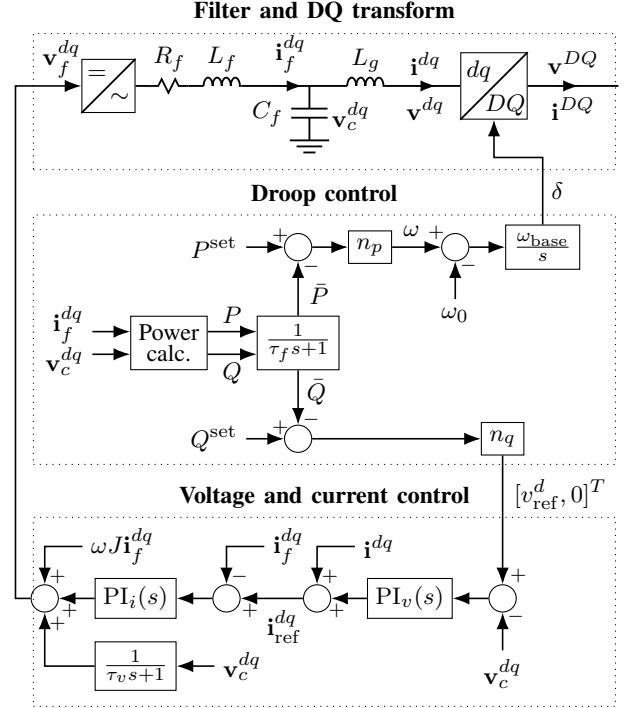

The parameters for the three grid-forming inverters in the case study of \cref{sec:case_study} are listed below, where values are given in the order $\{\nu_1, \nu_2, \nu_3\}$.
Filter resistance $R_f = \{0.005, 0.01, 0.008\}$~\si{\pu},
filter inductance $L_f = \{0.05, 0.06, 0.05\}$~\si{\pu},
filter capacitance $C_f = \{0.06, 0.06, 0.06\}$~\si{\pu},
grid-side inductance $L_g = \{0.15, 0.12, 0.1\}$~\si{\pu},
power filter time constant $\tau_f = \{0.008, 0.009, 0.007\}$~\SI{}{s},
active power droop gain $n_p = \{0.01, 0.01, 0.02\}$,
reactive power droop gain $n_q = \{0.02, 0.01, 0.01\}$,
voltage controller proportional gain $k_{vp} = \{2, 1.5, 1.75\}$,
voltage controller integral gain $k_{vi} = \{10, 8, 8\}$,
current controller proportional gain $k_{ip} = \{0.3, 0.5, 0.4\}$,
current controller integral gain $k_{ii} = \{10, 8, 9\}$,
and voltage feedforward filter time constant $\tau_v = \{0.03, 0.02, 0.02\}$~\SI{}{s}.

\section*{References}
\bibliographystyle{IEEEtran}
\bibliography{./Files/references.bib}
\end{document}

%% file: Files/preamble.tex
\usepackage[dvipsnames]{xcolor}
\definecolor{subsectioncolor}{rgb}{0,0,0}
\usepackage{generic}
\def\BibTeX{{\rm B\kern-.05em{\sc i\kern-.025em b}\kern-.08em
    T\kern-.1667em\lower.7ex\hbox{E}\kern-.125emX}}
\usepackage{amsmath,amssymb,amsfonts}
\usepackage{algorithmic}
\usepackage{algorithm,algorithmic}
\usepackage{textcomp}
\usepackage{bbold}
\usepackage{url}
\usepackage{cite}
\allowdisplaybreaks
\usepackage[hidelinks]{hyperref}
\usepackage[noabbrev,capitalise,nameinlink]{cleveref}
\usepackage{nicefrac, xfrac}
\usepackage[draft]{changes}
\usepackage{enumerate}

\usepackage{graphicx}
\graphicspath{ {./Figures/} }
\usepackage{subcaption}

\usepackage{siunitx}
\DeclareSIUnit\year{yr}
\DeclareSIUnit{\pu}{p.u.}
\NewDocumentCommand\SIpi{}{\text{\ensuremath{\pi}}}
\usepackage{aliascnt}
\newtheorem{theorem}{Theorem}

\newaliascnt{proposition}{theorem}
\newtheorem{proposition}[proposition]{Proposition}
\aliascntresetthe{proposition}
\crefname{proposition}{Proposition}{Propositions}

\newaliascnt{corollary}{theorem}
\newtheorem{corollary}[corollary]{Corollary}
\aliascntresetthe{corollary}
\crefname{corollary}{Corollary}{Corollaries}

\newaliascnt{lemma}{theorem}
\newtheorem{lemma}[lemma]{Lemma}
\aliascntresetthe{lemma}
\crefname{lemma}{Lemma}{Lemmas}

\newtheorem{definition}{Definition}
\newtheorem{remark}{Remark}
\newtheorem{assumption}{Assumption}
\crefname{assumption}{Assumption}{Assumptions}

\usepackage[short]{optidef}

\usepackage{tikz}
\usetikzlibrary{positioning}
\usetikzlibrary{matrix}
\usetikzlibrary {arrows.meta}
\usetikzlibrary{shapes}
\usetikzlibrary{calc}
\usetikzlibrary{math}
\usetikzlibrary{decorations.markings}
\usetikzlibrary{angles,quotes}
\usepackage[european,straightvoltages]{circuitikz}
 
\usetikzlibrary{decorations.pathmorphing,arrows.meta}

%% file: Figures/tikz_nyquist_mod.tex
\begin{tikzpicture}[>=Latex, line cap=round,scale=0.4]
\small
\def\R{3.2}      
\def\eps{0.5}   

\tikzset{
  nyqpath/.style={
    very thick,
    postaction={decorate},
    decoration={markings,
      mark=at position 0.13 with {\arrow{>}},
      mark=at position 0.34 with {\arrow{>}},
      mark=at position 0.55 with {\arrow{>}},
      mark=at position 0.9  with {\arrow{>}}
    }
  }
}

\draw[->] (-1.0,0) -- (\R+0.9,0) node[below right] {$\Re$};
\draw[->] (0,-\R-0.9) -- (0,\R+0.9) node[above left] {$\Im$};

\node[above left] at (0,\R) {$j\infty$};
\node[below left] at (0,-\R) {$-j\infty$};

\draw[nyqpath]
  (0,-\R)
    -- (0,-\eps)
    arc[start angle=-90, end angle=90, radius=\eps cm]
    -- (0,\R)
    arc[start angle=90, end angle=-90, radius=\R cm];

\node at (\R*1.2,0.5) {$\Gamma_R$};
\node at (2*\eps,0.5) {$\Gamma_{\epsilon}$};
\node at (-0.75,\R*0.75) {$\Gamma_{j\omega}^+$};
\node at (-0.75,-\R*0.75) {$\Gamma_{j\omega}^-$};


\end{tikzpicture}

%% file: Figures/tikz_model_iv.tex
\begin{tikzpicture}[scale=1]
\small

\node[matrix, draw,
    row sep=1pt, column sep=1pt,
    inner sep=1pt]
    (sys_Z) at (0,0)
    {
        \node{$Z_1(s)$}; & & \\[-8pt]
        & \node{$\ddots$}; & \\
        & & \node{$Z_{N_B}(s)$}; \\
    };
    \node at (sys_Z.south east)[right=3pt,yshift=6pt]{$Z_B(s)$};
\node[matrix, draw=black,
    below=of sys_Z,yshift=+20pt,
    row sep=1pt, column sep=1pt,
    inner sep=1pt]
    (sys_Yp) 
    {
        \node{$Y_{ij}(s)$}; & & \\[-8pt]
        & \node{$\ddots$}; & \\
        & & \node{$Y_{pq}(s)$}; \\
    };
    \node at (sys_Yp.south east)[right=3pt,yshift=1pt]{$Y_P(s)$};
\node[matrix, draw=black,
    below=of sys_Yp,yshift=+20pt,row sep=1pt, column sep=1pt,
    inner sep=1pt]
    (sys_Y0) 
    {
        \node{$Y_1(s)$}; & & \\[-8pt]
        & \node{$\ddots$}; & \\
        & & \node{$Y_{N_B}(s)$}; \\
    };
    \node at (sys_Y0.south east)[right=3pt,yshift=6pt]{$Y_0(s)$};

\node[draw, circle, minimum size=0.75cm,
            left=of sys_Yp,xshift=20pt]
            (B)  
            {$\mathcal{B}_N$};
\node[draw, circle, minimum size=0.75cm,
            right=of sys_Yp,xshift=-20pt]
            (Bt)   {$\mathcal{B}_N^T$};

\node[draw, circle, minimum size=0.05cm,
    left=of B,xshift=15pt]
    (left1) 
    {};
    \node[yshift=5pt,xshift=5pt] at (left1.north east){$+$};
    \node[yshift=-5pt,xshift=5pt] at (left1.south east){$+$};
\node[right= of Bt,xshift=-20pt] (right1){};
\node[right= of right1,xshift=-15pt] (right2){};
\node[draw, circle, minimum size=0.05cm]
    (neg) at (-4,0)
    {};
    \node[yshift=-5pt,xshift=5pt] at (neg.south east){$-$};

\node[draw,black,thin,dotted, anchor=north west,
    minimum width =7cm, minimum height =3.1cm]
    (box) at (-3.6,-0.8) {};
    \node at (box.south east)[right=3pt,yshift=6pt]{$Y_N(s)$};

\draw[-stealth] (neg.east) -- (sys_Z.west)
    node[midway,above]{$-\Delta i_B^{DQ}$};
\draw[] (sys_Z.east) -| (right2.center)
    node[pos=0.25,above]{$\Delta v_B^{DQ}$};
\draw[-stealth] (right2.center) -- (Bt.east);
\draw[-stealth] (right1.center) |- (sys_Y0.east);
\draw[-stealth] (Bt.west) -- (sys_Yp.east);
\draw[-stealth] (sys_Yp.west) -- (B.east);
\draw[-stealth] (B.west) -- (left1.east);
\draw[-stealth] (sys_Y0.west) -| (left1.south);
\draw[-stealth] (left1.west) -| (neg.south);

\end{tikzpicture}

%% file: Figures/tikz_extended_graph.tex
\begin{tikzpicture}[scale=1,
x=2.2cm, y=1.3cm]
    \small
    
    \node[rectangle, draw, thick, minimum size=0.8cm] 
        (bus_1) at (0,0) {$Z_1(s)$};
    \node[rectangle, draw, thick, minimum size=0.8cm] 
        (bus_2) at (1,0) {$Z_2(s)$};
    \node[rectangle, draw, thick, minimum size=0.8cm] 
        (bus_3) at (2,0) {$Z_3(s)$};
    \node[rectangle, draw, thick, minimum size=0.8cm] 
        (bus_4) at (3,0) {$Z_4(s)$};
    \node[ground] at (0,-1) {};
    \node[ground] at (1,-1) {};
    \node[ground] at (2,-1) {};
    \node[ground] at (3,-1) {};
    \draw[stealth-] (bus_1) -- (bus_2)
        node[midway,above]{$Y_{12}(s)$};
    \draw[stealth-]  (bus_2) -- (bus_3)
        node[midway,above]{$Y_{23}(s)$};
    \draw[stealth-]  (bus_3) -- (bus_4)
        node[midway,above]{$Y_{34}(s)$};

    \draw (0,-1) -- (bus_1)
        node[midway,above,sloped]{$Y_1(s)$};
    \draw (1,-1) -- (bus_2)
        node[midway,above,sloped]{$Y_2(s)$};
    \draw (2,-1) -- (bus_3)
        node[midway,above,sloped]{$Y_3(s)$};
    \draw (3,-1) -- (bus_4)
        node[midway,above,sloped]{$Y_4(s)$};

    \node[draw, rectangle, dotted, anchor=north west,
			minimum width=3.1cm, 	minimum height=2cm, label=$E_2(s)$]
			(sigma) at (0.3,0.5) {};

\end{tikzpicture}

%% file: Figures/tikz_extended_block_diag.tex
\small
\[
Z_{E,i}(s) = \begin{bmatrix}
    Z_2(s) & -Z_2(s) & Z_2(s)  \\
    -Z_2(s) & Z_2(s) & -Z_2(s) \\
    Z_2(s) & -Z_2(s) & Z_2(s)\\
\end{bmatrix}
\]
\[
Y_{E,i}(s) = \begin{bmatrix}
    Y_{12}(s) & & \\
    & Y_{23}(s) & \\
    & & Y_2(s) 
\end{bmatrix}
\]

%% file: Figures/9bus_diagram.tex
\small
\begin{tikzpicture}[
    scale=0.3,
    bus/.style={line width=2.2pt},
    trafo/.style={decorate,decoration={zigzag,pre length=2mm,post length=2mm,
                  segment length=1.6mm,amplitude=1mm}},
    gen/.style={circle,draw,minimum size=6mm,inner sep=0pt},
    busnum/.style={circle,draw,inner sep=1.2pt,minimum size=4.5mm},
    load/.style={-{Latex[length=3mm]},line width=1pt},
    lbl/.style={font=\small},
    every node/.style={font=\small}
  ]

  \draw (0,10) -- (11,10);
  \draw[bus] (0,9.0)  -- (0,10.8);   
  \draw[bus] (5.5,9.4) -- (5.5,11.2);
  \draw[bus] (11,9.0) -- (11,10.8);  


  \draw[load] (5.5,10.8) -- ++(2.5,0) node[right,above] {Load C};
  \node[busnum] at (5.5,8.3) {8};

  \node[gen] (INF) at (0.4,12.8) {$\sim$};
  \draw (INF) -- (2.5,12.8);
  \fill (2.5,12.8) circle (0.18);                
  \draw (2.5,12.8) -- ++(35:1.8);                
  \fill (4.3,12.8) circle (0.18);                
  \draw (4.3,12.8) -- (5.5,12.8) -- (5.5,11.2);

  \node[gen] (G2) at (-5.3,10) {2};
  \draw (G2) -- (-4.2,10);
  \draw[trafo] (-4.2,10) -- (0,10);              

  \node[gen] (G3) at (16.3,10) {3};
  \draw (G3) -- (15.2,10);
  \draw[trafo] (11,10) -- (15.2,10);             

  \draw (0,9.2) -| (2.2,2.0);
  \node[busnum] at (0.0,8) {7};
  \draw[bus] (0.5,5.6) -- (2.8,5.6);
  \node[busnum] at (4,5.6) {5};
  \draw[load] (1.0,5.6) -- ++(0,-2) node[below,left] {Load A};

  \draw (11,9.2) -| (8.8,2.0);
  \node[busnum] at (11,8.0) {9};
  \draw[bus] (8.2,5.6) -- (10.5,5.6);
  \node[busnum] at (7,5.6) {6};
  \draw[load] (10.0,5.6) -- ++(0,-2) node[below,right] {Load B};

  \draw[bus] (1.2,2.0) -- (9.8,2.0);             
  \node[busnum] at (11,2.0) {4};
  \draw[trafo] (5.5,2.0) -- (5.5,-0.6);          
  \node[gen] (G1) at (5.5,-1.7) {1};
  \draw (5.5,-0.6) -- (G1);

\end{tikzpicture}

%% file: Figures/inv_block_diagram.tex
\begin{tikzpicture}[
    blk/.style={rectangle,draw,minimum width=0.5cm,minimum height=0.5cm,
                align=center,fill=white,font=\small},
    gain/.style={rectangle,draw,minimum width=0.5cm,minimum height=0.4cm,
                 align=center,fill=white,font=\small},
    sum/.style={circle,draw,inner sep=0pt,minimum size=3.6mm},
    wire/.style={line width=0.7pt},
    sig/.style={-{Latex[length=2.2mm]},line width=0.6pt},
    res/.style={decorate,decoration={zigzag,pre length=1mm,post length=1mm,
                segment length=1.8mm,amplitude=0.9mm}},
    ind/.style={decorate,decoration={bumps,segment length=2.4mm,amplitude=1.1mm}},
    dot/.style={circle,fill,inner sep=0pt,minimum size=1.6mm},
    lbl/.style={font=\small},
    sgn/.style={font=\scriptsize,inner sep=1.5pt},
    every node/.style={font=\small}
  ]
\small
  \node[blk,minimum width=20,minimum height=20] (inv)   at (0,0)       {};
  \node[blk,minimum width=25,minimum height=25]   (tblk)  at (5.1,0)     {};
  \node[blk]                                            (pcalc) at (0.8,-3.4)  {\small Power\\[-2pt] calc.};
  \node[blk,minimum width=0.5cm,minimum height=0.5cm,
  right=of pcalc,xshift=-10]  
    (lpf) {$\tfrac{1}{\tau_f s + 1}$};
  \node[sum,
    above=of lpf,yshift=-8] 
    (sump)  {};
 \node[sum,
    below=of lpf,yshift=8] 
    (sumq)  {};
  \node[gain,
    right=of sump,xshift=-15]
    (np)  {$n_p$};
  \node[sum,
    right=of np,xshift=-10] 
    (sumw)  {};
  \node[gain,
    right=of sumw,xshift=-15]
    (int) {$\tfrac{\omega_{\mathrm{base}}}{s}$};
  \node[gain,
    right=of sumq,xshift=35]
    (nq)  {$n_q$};
  \node[sum,
    below=of nq,yshift=-20]           
    (sumv) {};
  \node[blk,minimum width=0.5cm,
    left=of sumv,xshift=15]
    (piv)  {$\mathrm{PI}_v(s)$};
  \node[sum,
    left=of piv,xshift=15]
    (sumi)  {};
  \node[sum,
    left=of sumi,xshift=5]
    (sumi2) {};
  \node[blk,minimum width=0.5cm,
    left=of sumi2,xshift=15]
    (pii) {$\mathrm{PI}_i(s)$};
  \node[sum,
    left=of pii,xshift=15]
    (sum3) {};
  \node[blk,minimum width=0.5cm,minimum height=0.5cm,
  below=of pii,xshift=0,yshift=20]  
  (ff)   {$\tfrac{1}{\tau_{v}s + 1}$};

  \coordinate (c0) at (inv.east);

  \draw[wire] (inv.south west) -- (inv.north east);
  \node at ($(inv.center)+(-0.15,0.15)$) {$=$};
  \node at ($(inv.center)+(0.15,-0.15)$) {$\sim$};

  \draw[wire] (c0) -- ++(0.2,0);
  \draw[wire,res] ($(c0)+(0.2,0)$)  -- ++(0.5,0) node[midway,above=2pt,lbl] {$R_f$};
  \draw[wire]     ($(c0)+(0.7,0)$)  -- ++(0.2,0);
  \draw[wire,ind] ($(c0)+(0.9,0)$) -- ++(0.5,0) node[midway,above=2pt,lbl] {$L_f$};
  \draw[wire]     ($(c0)+(1.4,0)$) -- ++(1.4,0);

  \draw[-{Latex[length=2mm]}] ($(c0)+(1.9,0)$) -- ++(0.3,0)
        node[midway,above=2pt,lbl] {$\mathbf{i}_f^{dq}$};

  \coordinate (capX) at ($(c0)+(2.3,0)$);
  \draw[wire] (capX) -- ++(0,-0.3);
  \draw[wire] ($(capX)+(-0.24,-0.3)$)  -- ++(0.48,0)
        node[left=1pt,yshift=-2,pos=0.1,lbl] {$C_f$};
  \draw[wire] ($(capX)+(-0.24,-0.42)$) -- ++(0.48,0);
  \draw[wire] ($(capX)+(0,-0.42)$) -- ++(0,-0.3);
  \draw[wire] ($(capX)+(-0.26,-0.72)$)  -- ++(0.52,0);
  \draw[wire] ($(capX)+(-0.17,-0.8)$) -- ++(0.34,0);
  \draw[wire] ($(capX)+(-0.08,-0.88)$) -- ++(0.16,0);

  \coordinate (vcTap) at ($(capX)+(0.6,-0.7)$);  
  \node[lbl] at ($(vcTap)+(-0.05,0.3)$) {$\mathbf{v}_c^{dq}$};
  \draw[-{Latex[length=2mm]}] ($(c0)+(3.7,0)$) -- ++(0.27,0)
        node[midway,above=1pt,lbl] {$\mathbf{i}^{dq}$}
        node[midway,below=1pt,lbl] {$\mathbf{v}^{dq}$};

  \draw[wire,ind] ($(c0)+(2.8,0)$) -- ++(0.5,0) node[midway,above=2pt,lbl] {$L_g$};
  \draw[wire]     ($(c0)+(3.3,0)$) -- (tblk.west);

  \draw[wire] (tblk.south west) -- (tblk.north east);
  \node at ($(tblk.center)+(-0.2,0.2)$) {$dq$};
  \node at ($(tblk.center)+(0.15,-0.25)$) {$DQ$};

  \draw[wire] (tblk.east) -- ++(1.2,0)
        node[midway,above=1pt,lbl] {$\mathbf{v}^{DQ}$}
        node[midway,below=1pt,lbl] {$\mathbf{i}^{DQ}$};
  \draw[-{Latex[length=2mm]}] ($(tblk.east)+(0.4,0)$) -- ++(0.27,0);

  \coordinate (pcVc) at ([yshift=1.5mm]pcalc.west);
  \coordinate (pcIf) at ([yshift=-1.5mm]pcalc.west);

  \draw[sig]  ($(pcVc)+(-0.5,0.0)$)  -- (pcVc)
    node[pos=0.0,left,lbl,yshift=2.5] {$\mathbf{i}_f^{dq}$};
  \draw[sig]  ($(ff.east)+(0.5,0)$) -- (ff.east)
    node[pos=0.0,right,lbl] {$\mathbf{v}_c^{dq}$};

  \draw[sig]  ($(pcIf)+(-0.5,0.0)$) -- (pcIf)
    node[pos=0.0,left,lbl,yshift=-2.5] {$\mathbf{v}_c^{dq}$};

  \draw[sig] ([yshift= 1.5mm]pcalc.east) -- ([yshift= 1.5mm]lpf.west)
        node[midway,above=0pt,lbl] {$P$};
  \draw[sig] ([yshift=-1.5mm]pcalc.east) -- ([yshift=-1.5mm]lpf.west)
        node[midway,below=0pt,lbl] {$Q$};

  \draw[sig] ($(lpf.north)+(0.0,0)$) -- (sump.south)
    node[pos=0.4,right=1pt,lbl] {$\bar{P}$}
    node[pos=0.9,right=0pt,sgn] {$-$};
  \draw[sig]  ($(sump.west)+(-0.5,-0.0)$)  -- (sump.west)
    node[pos=0.0,left,lbl] {$P^{\mathrm{set}}$}
    node[pos=0.9,above,sgn] {$+$};
  \draw[sig] (sump.east) -- (np.west);
  \draw[sig] (np.east) -- (sumw.west)
        node[pos=0.4,above=1pt,lbl] {$\omega$}
        node[pos=0.9,above=1pt,sgn] {$+$};
   \draw[sig]  ($(sumw.south)+(0.0,-0.5)$)  -- (sumw.south)
    node[pos=0.0,below,lbl] {$\omega_0$}
    node[pos=0.9,right,sgn] {$-$};
  \draw[sig] (sumw.east) -- (int.west);
  \draw[sig] (int.north) |- ($(tblk.south) + (0,-0.5)$) 
    node[pos=0.25,right,lbl] {$\delta$}
  -- (tblk.south) ;

  \draw[sig] ($(lpf.south)+(0.0,0)$) -- (sumq.north)
        node[pos=0.4,right,lbl] {$\bar{Q}$}
        node[pos=0.9,right,sgn] {$-$};
  \draw[sig]  ($(sumq.west)+(-0.5,0)$)  -- (sumq.west)
    node[pos=0.0,left,lbl] {$Q^{\mathrm{set}}$}
    node[pos=0.9,above,sgn] {$+$};
  \draw[sig] (sumq.east) -- (nq.west);
  \draw[sig] (nq.south) -- (sumv.north)
        node[pos=0.3,right=1pt,lbl] {$[v_{\mathrm{ref}}^d,0]^T$}
        node[pos=0.93,right=1pt,sgn] {$+$};
    \draw[sig]  ($(sumv.south)+(0,-0.5)$)  -- (sumv.south)
    node[pos=0.0,below,lbl] {$\mathbf{v}_c^{dq}$}
    node[pos=0.9,right,sgn] {$-$};

  \draw[sig] (sumv) -- (piv.east);
  \draw[sig] (piv.west) -- (sumi.east)
        node[sgn,pos=0.85,below] {$+$};
  \draw[sig] (sumi) -- (sumi2.east)
        node[pos=0.3,below=1pt,lbl] {$\mathbf{i}_{\mathrm{ref}}^{dq}$}
        node[sgn,pos=0.85,below] {$+$};
   \draw[sig] ($(sumi.north)+(0.5,0.5)$)  -| (sumi.north)        
        node[pos=0.90,right,sgn] {$+$}
        node[pos=0.0,right,lbl] {$\mathbf{i}^{dq}$};
   \draw[sig]  ($(sumi2.north)+(0.5,0.5)$) -| (sumi2.north)         
        node[pos=0.0,right,lbl] {$\mathbf{i}_f^{dq}$}
        node[pos=0.90,right,sgn] {$-$};

  \draw[sig] (sumi2) -- (pii.east);
  \draw[sig] (pii.west) -- (sum3.east)
        node[sgn,pos=0.7,below] {$+$};
  \draw[sig] (ff.west) -| (sum3.south)
        node[sgn,pos=0.9,right] {$+$};
   \draw[sig]  ($(sum3.north)+(0.5,0.5)$)  -| (sum3.north)
    node[pos=0.0,right,lbl] {$\omega J\mathbf{i}_f^{dq}$}
    node[sgn,pos=0.9,right] {$+$};

  \draw[sig] (sum3.west) -- ($(sum3.west)+(-0.2,0)$) |- (inv.west)
        node[pos=0.85,above,lbl] {$\mathbf{v}_f^{dq}$};

\node[draw,black,thin,dotted, anchor=north west,
    minimum width =7.7cm, minimum height =1.8 cm]
    (box1) at (-1.0,0.7) {};
    \node at (box1.north)[above=1pt]{\small \textbf{Filter and DQ transform}};

\node[draw,black,thin,dotted, anchor=north west,
    minimum width =7.7cm, minimum height =3.3cm]
    (box2) at (-1.0,-1.7) {};
    \node at (box2.north)[above=0pt]{\small \textbf{Droop control}};

\node[draw,black,thin,dotted, anchor=north west,
    minimum width =7.7cm, minimum height =2.5cm]
    (box3) at (-1.0,-5.7) {};
    \node at (box3.north)[above=0pt]{\small \textbf{Voltage and current control}};

\end{tikzpicture}